\documentclass[11pt]{article}

\usepackage[round,authoryear]{natbib}
\usepackage[margin=2.5cm]{geometry}
\usepackage{amsmath,amssymb,amsthm}
\usepackage{array}
\usepackage{microtype}
\usepackage{xcolor}
\usepackage{tcolorbox}
\usepackage[colorlinks=true,linkcolor=blue!55!black,citecolor=blue!55!black,urlcolor=blue!55!black]{hyperref}

\hypersetup{
  pdftitle={Sharper Gaussian Covers in the Bansal--Huang--Lee Coloring Framework},
  pdfsubject={Approximation algorithms, Gaussian inequalities, and graph coloring}
}

\newtheorem{theorem}{Theorem}[section]
\newtheorem{lemma}[theorem]{Lemma}
\newtheorem{proposition}[theorem]{Proposition}

\theoremstyle{definition}
\newtheorem{definition}[theorem]{Definition}
\theoremstyle{remark}
\newtheorem{remark}[theorem]{Remark}

\newtcolorbox{mainresultbox}{
  colback=black!4,
  colframe=black!45,
  boxrule=0.6pt,
  arc=1mm,
  left=5pt,
  right=5pt,
  top=4pt,
  bottom=4pt
}

\renewcommand{\Pr}{\mathbb P}
\newcommand{\tail}{\overline\Phi}

\title{Sharper Gaussian Covers in the\\
Bansal--Huang--Lee Coloring Framework}
\author{Konstantin Makarychev\\Northwestern University \and Yury Makarychev\\TTIC}
\date{}

\begin{document}
\maketitle

\begin{abstract}
\citet*{BHL26} recently gave a polynomial-time algorithm that colors every
$3$-colorable graph on $n$ vertices with $O(n^{0.19539})$ colors.  We
improve their bound to $O(n^{0.17794})$ colors.

The improvement comes from a sharper rule for combining Gaussian covers.
We prove the rule from Ehrhard's inequality by comparing the probability
that a Gaussian vector misses a polyhedron at two thresholds.  At the
farther threshold a union bound records how many halfspaces define the
polyhedron, and this information reduces the loss in the combination.  We
then organize the successive neighborhood steps of the Bansal--Huang--Lee
argument into a single recursion.  The smaller loss lets the recursion run
one step farther than before, and at that step it would require more
distinct vertices than the graph contains.  This rules out the remaining
case.  Combining the resulting bounded-degree guarantee with the
dense-graph algorithm of \citet*{KTY24} gives a randomized polynomial-time
algorithm that colors every $3$-colorable graph on $n$ vertices with
$O(n^{0.17794})$ colors.
\end{abstract}

\paragraph{AI Use Acknowledgment}
The proof was developed through discussions with ChatGPT (GPT-5.6 Pro). It was subsequently revised and polished by the authors. The authors take full responsibility for the correctness of the proof.

\section{Introduction}

We design a polynomial-time approximation algorithm for graph $3$-coloring.
Given an $n$-vertex graph promised to be $3$-colorable, our algorithm finds a
proper coloring using $O(n^{0.17794})$ colors.

\citet*{Karp72} proved that the general chromatic-number decision problem is
NP-complete; in that problem, the allowed number of colors is part of the
instance.  A proper $k$-coloring of a graph assigns one of $k$ colors to each
vertex so that the endpoints of every edge receive different colors.
A graph is $k$-colorable if such an assignment exists.  The exact
$3$-coloring problem asks whether a given graph is $3$-colorable.  Here we
study the approximation variant: assuming that the graph is $3$-colorable,
find a proper coloring using as few colors as possible.

\citet*{Stockmeyer73} subsequently proved that the problem remains
NP-complete when the number of colors is fixed to three, even for planar
graphs.  \citet*{GJS74} strengthened this result to planar graphs of maximum
degree at most four.
\citet*{KLS00} proved that coloring a $3$-colorable graph with four colors
remains NP-hard, and \citet*{BBKO21} later proved the same for five colors.
Thus, unless $\mathbf P=\mathbf{NP}$, no polynomial-time algorithm can
properly color every $3$-colorable graph using only five colors.  In contrast,
the best known algorithms still require a polynomial number of colors.  This
large gap has motivated a long sequence of increasingly geometric ideas.

\citet*{Wig83} provided the first nontrivial approximation algorithm, giving
a combinatorial method that uses $O(n^{1/2})$ colors.  \citet*{BR90} gave a proof of an
$O(\sqrt{n/\log n})$-coloring algorithm; as they noted, the same
improvement was obtained independently, using different methods, by
Linial, Saks, and Wigderson and by Raghavan.  \citet*{Blum94} obtained algorithms using
$\widetilde O(n^{2/5})$ and $\widetilde O(n^{3/8})$ colors by examining
neighborhoods and common neighborhoods to find large independent sets or
directly color substantial parts of the graph.  \citet*{KMS98} introduced a
vector-coloring semidefinite programming (SDP) relaxation, closely related to
the Lov\'asz theta function of \citet*{Lovasz79}, and showed how to round this
relaxation to color $3$-colorable graphs.  Combined with the degree reduction
of \citet*{Wig83}, their
approach gives $O(n^{1/4}\sqrt{\log n})$ colors.  \citet*{BK97} showed how to
combine the combinatorial and semidefinite approaches to achieve
$\widetilde O(n^{3/14})$ colors.  Here $\widetilde O$ suppresses
polylogarithmic factors.

A major conceptual advance came from \citet*{ACC06}.  In addition to improving
the bound to $\widetilde O(n^{0.2111})$, they introduced a new way to exploit
the case in which KMS rounding alone makes little progress.  \citet*{Chlamtac07}
developed these ideas further and improved the bound to
$\widetilde O(n^{0.2072})$.  This line of work showed that the geometry of
second-level neighborhoods contains additional structure that can still be
used algorithmically.  This insight is a central precursor to the present
paper.  \citet*{KT17} then
broke the $n^{1/5}$ barrier, showing how to color a $3$-colorable graph with
$\widetilde O(n^{0.19996})$ colors.  In subsequent work, \citet*{KTY24}
achieved $\widetilde O(n^{0.19747})$.  Very recently, \citet*{BHL26} obtained
$O(n^{0.19539})$ colors.  Our main result improves this bound as follows.

\begin{mainresultbox}
\begin{theorem}
\label{thm:final-coloring}
There is a randomized polynomial-time algorithm that properly colors every
$3$-colorable graph on $n$ vertices using $O(n^{0.17794})$ colors.
\end{theorem}
\end{mainresultbox}

\paragraph{Where the improvement comes from.}
We first recall the vector-coloring approach of \citet*{KMS98}.  The KMS
algorithm solves an SDP that assigns a unit vector to every vertex so that
the vectors of adjacent vertices have inner product at most $-1/2$.  The
intended solution assigns each vertex one of three unit vectors forming a
regular triangle, one vector for each color.  A solution returned by an SDP
solver, however, can look very different from this intended solution.

The KMS algorithm chooses a random Gaussian vector $g$ and a threshold $t$
that depends on the maximum degree of the graph.  It selects the vertices
whose vectors have inner product at least $t$ with $g$, and then discards
every selected vertex that has a selected neighbor.  The surviving vertices
form an independent set.  The algorithm assigns a new color to this set,
removes it from the graph, and repeats.  Sufficiently large independent sets
therefore lead to a valid coloring.

Later work asked what can be learned when this rounding does not produce a
large independent set.  \citet*{ACC06} and, subsequently,
\citet*{Chlamtac07,Chlamtac09} showed that the discarded vertices are not
merely a loss.  Their presence forces a useful geometric pattern among the
vertices reached by two-edge walks.  This pattern is described using
Gaussian covers: a collection of vectors is a Gaussian cover if a random
Gaussian vector has a noticeable chance of having a large inner product with
at least one vector in the collection.  When the rounding discards many
vertices, the directions from many vertices toward their neighbors form
such covers, and these covers can be combined along walks to find a large
independent set farther away.  \citet*{BHL26} sharpened this machinery,
carried it from two-edge walks to three-edge walks, and used a stronger SDP
relaxation that retains information about small groups of vertices.
Combining the result with the dense-graph routine of \citet*{KTY24},
through the framework of \citet*{BK97} and \citet*{KT17}, gave their
$O(n^{0.19539})$ bound.

Our contribution has two parts.  The first is a sharper rule for combining
two Gaussian covers, Theorem~\ref{thm:two-threshold}.  It is a variant of
Theorem~3.4.4 of \citet*{Chlamtac09}, restated as Lemma~4.14 of
\citet*{BHL26}.  Like those results, it combines an outer cover with an
inner cover attached to each outer vector.  The earlier rule estimated the
chance that an inner cover fails using only its probability at the original
threshold.  Our rule also looks at a second, farther-out threshold, where a
union bound remembers how many vectors the inner cover contains and how long
they are.  Matching this estimate to the number of outer vectors gives a
smaller loss in the combined threshold.  We prove the rule in full from
Ehrhard's inequality in Section~\ref{sec:two-threshold}.

The second part is organizational.  We write the successive neighborhood
steps as a single recursion that tracks three numbers: the geometry of the
current endpoints, the threshold of their cover, and an upper bound on how
many of them there are.  The smaller loss lets this recursion run to a
fourth neighborhood level.  At levels two and three, the algorithm compares
the number of reachable endpoints with a cutoff fixed in advance.  Above the
cutoff, KMS rounding yields the desired independent set.  Below it, the
size bound feeds the next step of the recursion.  If both levels are below
their cutoffs, the fourth-level cover would need more than $n$ distinct
endpoints, which is impossible.  This contradiction closes the last case
without another rounding routine.

\paragraph{Related work.}
Ehrhard's inequality is a Gaussian analogue of the Brunn--Minkowski inequality
and is closely tied to Gaussian isoperimetry.  It has also found several uses
in the study of Gaussian random variables.  \citet*{Bobkov08} observed that
Ehrhard concavity characterizes the possible distributions of suprema of
Gaussian processes.  \citet*{PV18} used a
tangent-line argument to obtain variance-sensitive lower-tail bounds for
convex Gaussian functionals, with applications to small-ball probabilities
for Gaussian processes.  Building on this viewpoint, \citet*{Valettas19}
used Ehrhard-based estimates to study when the usual Gaussian concentration
bound for a convex function is tight.  \citet*{Chen26} used a chord argument
based on Ehrhard concavity to obtain a sharp upper moderate-deviation bound
for Gaussian maxima.  This argument compares probabilities at two thresholds
to control how quickly the probability can deteriorate as the threshold
changes.

\paragraph{Organization.}
The paper has two parts.  Section~\ref{sec:overview} is an extended
overview.  It explains every ingredient informally, with the computations
that show where the numbers come from, and it can be read on its own.
Sections~\ref{sec:preliminaries}--\ref{sec:coloring-proof} contain the
formal proof.  Section~\ref{sec:preliminaries} fixes notation and states
the results quoted from earlier work.  Section~\ref{sec:two-threshold}
proves the two-threshold inequality.  Section~\ref{sec:recursive-step}
proves the one-step recursion.  Section~\ref{sec:structural-input} shows
what a failed rounding gives us to start the recursion.
Section~\ref{sec:three-levels} runs the recursion through the three
transitions and proves the bounded-degree guarantee, and
Section~\ref{sec:coloring-proof} combines it with the dense-graph algorithm.
Appendix~\ref{sec:numerics} describes the computer-assisted verification of
the numerical inequalities, and Appendix~\ref{sec:kms-proof} proves the
rounding theorem in the uniform form we need.

\section{Extended overview}
\label{sec:overview}

This section explains the whole argument informally.  We keep the notation
of the formal part, but we allow ourselves to say ``about'' and to ignore
lower-order terms.  Throughout, $G=(V,E)$ is the input graph, $n=|V|$, and
$t\to\infty$ is a Gaussian threshold that will be about $\sqrt{\log n}$.
The reader who wants the precise statements can jump to the sections named
in each subsection.

At a glance, the proof has the following shape.  Gaussian rounding either
finds the required independent set directly or fails.  When it fails, it
exposes, around many vertices, weighted families of directions toward
neighbors that behave like Gaussian covers.  Our two-threshold inequality
combines such families along walks of increasing length, and a single
recursion records how the threshold of the combined cover evolves.  At the
two intermediate stages, the algorithm lists the possible walk endpoints and
compares the size of the list with a cutoff fixed in advance.  A large list
can be rounded directly; a small list gives the size bound needed to
continue.  If both lists are small, the next cover would require more than
$n$ distinct endpoints.  The resulting bounded-degree guarantee is then
balanced against the known dense-graph routine.

\subsection{The semidefinite relaxation}
\label{subsec:overview-sdp}

\citet*{KMS98} introduced the following semidefinite relaxation, which
assigns a unit vector $v_i$ to each vertex $i\in V$:
\begin{mainresultbox}
\begin{align*}
    \|v_i\| &= 1
        && \text{for all } i\in V,\\
    \langle v_i,v_j\rangle &\leq -\frac{1}{2}
        && \text{for all } \{i,j\}\in E.
\end{align*}
\end{mainresultbox}
The intended integral solution uses three unit vectors forming the
vertices of an equilateral triangle centered at the origin, and assigns
each vertex the vector corresponding to its color. These three vectors
have pairwise inner product $-1/2$, so every proper $3$-coloring yields
such a solution.  In a general solution to the relaxation, however, the
vectors $v_i$ may point in many different directions.  When equality holds
on every edge, the solution is called a \emph{strict vector $3$-coloring}.

The basic relaxation is enough for the original KMS rounding, but the later
steps need joint information about a few vertices.  For this purpose,
following \citet*{BHL26}, we use the three-round Sum-of-Squares (SoS), or
Lasserre, relaxation; see also \citet*{Las01,Par00}.  Let
$\mathcal C=\{\mathsf R,\mathsf G,\mathsf B\}$ be the set of colors, and
let $x_{i,C}$ be a binary variable indicating that vertex $i$ has color
$C$.  A proper $3$-coloring satisfies
\begin{mainresultbox}
\begin{align*}
    \sum_{C\in\mathcal C} x_{i,C} &= 1
        && \text{for every } i\in V,\\
    x_{i,C}x_{j,C} &= 0
        && \text{for every } \{i,j\}\in E,\ C\in\mathcal C,\\
    x_{i,C} &\in \{0,1\}
        && \text{for every } i\in V,\ C\in\mathcal C.
\end{align*}
\end{mainresultbox}
The relaxation replaces products of a few such variables by inner products
of vectors.  For every set $S$ of at most $k=3$ vertex--color pairs, it has
a vector $y_S$ representing the corresponding assignment monomial.  Writing
$y_{(i,C)}$ for $y_{\{(i,C)\}}$, the vectors are required to satisfy the
following constraints; inner products not fixed by the displayed identities
are constrained only by positive semidefiniteness of the common Gram matrix.
\begin{mainresultbox}
\begin{align*}
    \|y_\varnothing\| &= 1,\\
    \sum_{C\in\mathcal C} y_{(i,C)} &= y_\varnothing
        && \text{for every } i\in V,\\
    \langle y_{(i,C)},y_{(i,C')}\rangle &= 0
        && \text{for every } i\in V,\ C,C'\in\mathcal C,\ C\ne C',\\
    \langle y_{(i,C)},y_{(j,C)}\rangle &= 0
        && \text{for every } \{i,j\}\in E,\ C\in\mathcal C,\\
    \langle y_{S_1},y_{S_2}\rangle
        &= \langle y_{S_3},y_{S_4}\rangle
        && \text{if } S_1\cup S_2=S_3\cup S_4,\
           |S_1\cup S_2|\leq k.
\end{align*}
\end{mainresultbox}
An inner product such as $\langle y_{(i,C)},y_{(j,C')}\rangle$ can be read
as the \emph{pseudo-probability} that vertices $i$ and $j$ receive colors
$C$ and $C'$ in a locally consistent random coloring.  For instance, the
last constraint forces every vector indexed by an inconsistent partial
coloring to vanish.  These pseudo-probabilities need not come from any
global distribution over colorings of the whole graph.

We may also impose color symmetry for free.  Permute the color labels in the
Gram matrix in all six possible ways and average the resulting matrices.
The constraints are linear, so the averaged matrix is still feasible, and it
is unchanged by every color permutation.  In particular, the three
single-color vectors at each vertex have equal squared norm $1/3$.  Following
\citet*{BHL26}, write
\[
    y_{(i,\mathsf R)}
    = \frac{1}{3}y_\varnothing + \frac{\sqrt{2}}{3}v_i,
\]
where $v_i\perp y_\varnothing$ and $\|v_i\|=1$.  For every edge
$\{i,j\}\in E$, the same-color constraint gives
\[
    \langle y_{(i,\mathsf R)},y_{(j,\mathsf R)}\rangle
    =\frac{1}{9}+\frac{2}{9}\langle v_i,v_j\rangle=0,
\]
and hence $\langle v_i,v_j\rangle=-1/2$.  Thus the vectors $v_i$ form a
strict vector $3$-coloring.

Beyond this, the only consequences of the stronger relaxation that we use
are two results of \citet*{BHL26}, stated as
Lemmas~\ref{lem:bhl-negative-coloring} and~\ref{lem:bhl-positive-coloring}
below.  After a root vertex $i$ is fixed, they build better vector colorings
for sets of vertices whose vectors have a prescribed correlation with
$v_i$.  Both constructions are explicit linear combinations of the vectors
$y_S$.

\subsection{Gaussian rounding, and what we learn when it fails}
\label{subsec:overview-rounding}

Sample a standard Gaussian vector $g$ and select
$S=\{i:\langle g,v_i\rangle\geq t\}$.  The original KMS procedure keeps the
isolated vertices of $G[S]$.  The variant used here finds a maximal matching
in $G[S]$ with a fixed tie-breaking rule, deletes the matched vertices, and
keeps the rest.  The remaining vertices are independent: if two of them
were adjacent, the matching would not have been maximal.

For a standard normal variable $Z$, write $\Phi(t)=\Pr[Z\leq t]$ and
$\tail(t)=\Pr[Z\geq t]$.  In the classical KMS analysis, $\tail(t)$ is
about $\Delta^{-1/3}$, where $\Delta$ bounds the maximum degree.  About
$n\Delta^{-1/3}$ vertices are selected, and because adjacent vectors have
inner product $-1/2$, two adjacent vertices are rarely selected together.
This leaves an independent set of size $\widetilde\Omega(n\Delta^{-1/3})$.
\citet*{BHL26} instead use a slightly lower threshold, chosen so that
$\tail(t)=\Delta^{-1/[3(1+c)]}$ for a small constant $c>0$.  More vertices
are selected.  If few of them are deleted, the algorithm immediately gets a
larger independent set.  The interesting case is the opposite one.

We say that the rounding \emph{fails} if, for at least half of the
vertices, the conditional probability of being deleted given being selected
is at least $1/2$; this is Definition~4.1 of \citet*{BHL26}.  Failure is a
property of the rounding distribution, not of one unlucky draw.  If the
rounding does not fail, the expected number of surviving vertices is at
least $n\tail(t)/4$, and polynomially many independent repetitions find an
independent set of size $\Omega(n\tail(t))$ with high probability.  We then
say that the rounding \emph{succeeds}.  Success and failure are the two
cases of the analysis.

Failure gives geometric information about neighboring vertices.  For an
edge $\{i,j\}$ in a strict vector $3$-coloring, define the unit direction
$v_{ij}\perp v_i$ by
\[
 v_j=-\frac12v_i+\frac{\sqrt3}{2}v_{ij}.
\]
If both endpoints of the edge are selected, then
\[
 \langle g,v_{ij}\rangle
 =\frac{2}{\sqrt3}\left(
   \langle g,v_j\rangle+\frac12\langle g,v_i\rangle\right)
 \geq\sqrt3t.
\]
So every matched edge $\{i,j\}$ certifies that $g$ has inner product at
least $\sqrt3t$ with the direction $v_{ij}$.  Failure says that matched
edges are common.  Hence, around a typical vertex $i$, the directions
$v_{ij}$ toward its neighbors form a Gaussian cover at threshold
$\sqrt3t$: a random $g$ has a noticeable chance of crossing this threshold
against at least one of them.

We will need more than the covering probability.  The proof restricts the
family of directions several times, keeping only directions with a
prescribed geometry, and it must know how much covering probability
survives each restriction.  For this reason we use a \emph{packing
measure}, a term from \citet*{BHL26}: a weight on the directions such that
the weight of any subfamily is at most the probability that some direction
in the subfamily crosses the threshold, while the total weight is not too
small.  The probability that a given edge is matched, divided by
$\tail(t)$, is such a weight.  Restricting a packing to a subfamily of
weight $p$ leaves a cover of probability at least $p$.

The probabilities in this argument are allowed to be small, but not too
small.  All Gaussian tail probabilities have the form $\exp[-\Theta(t^2)]$,
and a factor such as a power of $t$ or of $\log t$ changes only the
lower-order term in the exponent.  We call a probability
\emph{subexponential} if it is $\exp[-o(t^2)]$.  All covers and packings
in the proof have subexponential mass, and losing a factor
$(\log t)^{O(1)}$ keeps them subexponential.

The formal output of a failed rounding is
Proposition~\ref{prop:structural-input}: there is a root vertex $i$ whose
neighbor directions form a subexponential packing, each retained neighbor
$j$ carries its own packing of next-edge directions $v_{jk}$, each retained
endpoint $k$ carries a packing of the directions $v_{k\ell}$ toward its own
neighbors, and the correlations along these walks lie in prescribed ranges.
The recursion uses only this output.

\subsection{Combining two covers: the two-threshold inequality}
\label{subsec:overview-two-threshold}

The analytic heart of the paper is a rule for combining two successive
moves along a walk.  We have a cover describing the possible first moves,
and, attached to each first move, another cover describing the possible
next moves.  We need both to succeed for the same first move.

We begin with a single cover.  Let $g$ be a standard Gaussian vector, let
$X=\{x_1,\ldots,x_m\}$ be unit vectors, and write
$M_X(g)=\max_j\langle g,x_j\rangle$.  Saying that $X$ covers at threshold
$s$ with probability $\delta$ means $\Pr[M_X(g)\geq s]\geq\delta$.
Borell's inequality~\citep{Borell75} then bounds the probability that the
cover misses a lower threshold: for every $\rho\geq0$,
\begin{equation}
 \Pr[M_X(g)\leq s-\rho]
 \leq\Phi\!\left(\Phi^{-1}(1-\delta)-\rho\right)
 \approx e^{-\rho^2/2}.
 \label{eq:overview-borell}
\end{equation}
To see the geometry, let $A=\{g:M_X(g)\geq s\}$.  Moving a point by
Euclidean distance at most $\rho$ changes every inner product
$\langle g,x_j\rangle$ by at most $\rho$, so the $\rho$-neighborhood of $A$
is contained in $\{g:M_X(g)\geq s-\rho\}$, and Borell's inequality
lower-bounds the Gaussian measure of that neighborhood.  This argument uses
only the measure of $A$.  It does not use how many vectors $X$ has or how
they are arranged.

Now consider two covers.  An index $p$ represents a possible first move,
and the vectors $X=\{x_p\}$ form the \emph{outer cover}.  Once $p$ is
fixed, a family $Y_p=\{u_{pm}\}\subseteq x_p^\perp$ describes the possible
next moves and is the \emph{inner cover} attached to $x_p$.  Combining the
covers means proving that, with noticeable probability, some single index
$p$ works twice: $x_p$ crosses the outer threshold, and some vector in
$Y_p$ crosses the inner threshold.  In the graph, a successful pair
$(p,m)$ identifies a vertex two moves farther along the walk.

Here is the basic computation.  Let the outer cover have threshold
$\Lambda t$, and suppose that it has at most $\exp(\mu t^2/2)$ members.  A
union bound shows that a cover at threshold $\Lambda t$ with
subexponential probability needs at least $\exp(\Lambda^2t^2/2)$ members,
since each member crosses with probability about $\exp(-\Lambda^2t^2/2)$.
So $\mu\geq\Lambda^2$, and we call
\[
 \Sigma=\mu-\Lambda^2
\]
the \emph{surplus} of the outer cover: the outer family may be larger than
the minimum forced by its covering probability by a factor
$\exp(\Sigma t^2/2)$.  Now lower the inner threshold by an amount $a$, and
let $B_p$ be the event that the inner cover attached to $p$ misses the
lowered threshold.  Because $Y_p\subseteq x_p^\perp$, the event $B_p$ is
independent of the event $A_p$ that $x_p$ crosses.  Hence
\[
 \Pr[\exists p:\ A_p\cap B_p]
 \leq\sum_p\Pr[A_p]\Pr[B_p]
 \leq\exp\!\left(\frac{\mu t^2}{2}\right)
     \exp\!\left(-\frac{\Lambda^2t^2}{2}\right)
     \max_p\Pr[B_p]
 =\exp\!\left(\frac{\Sigma t^2}{2}\right)\max_p\Pr[B_p].
\]
If this is much smaller than the outer covering probability, then with
noticeable probability some $A_p$ occurs without $B_p$, which is exactly a
successful pair.  So we need
$\Pr[B_p]\leq\exp(-\Sigma t^2/2)$ for every $p$.

Borell's inequality~\eqref{eq:overview-borell} gives
$\Pr[B_p]\approx e^{-a^2/2}$, so it suffices to lower the inner threshold by
$a=\sqrt\Sigma\,t$.  This is the rule of \citet*{Chlamtac09} and
\citet*[Lemma~4.14]{BHL26}.  In their setting the outer cover has threshold
$\sqrt3t$ and at most $\tail(\sqrt3t)^{-(1+c)}\approx\exp(3(1+c)t^2/2)$
members, so $\Sigma=3c$, and the inner threshold drops by
$\sqrt{3c}\,t$ times the length of the inner vector.

Our rule uses the fact that the inner cover has few members.  Suppose the
inner family has at most $\exp((1+\theta)S^2/2)$ members, where $S$ is its
threshold.  A union bound over the members shows that the inner maximum
exceeds a farther threshold $S+T$ with probability at most
\[
 \exp\!\left(\frac{(1+\theta)S^2}{2}\right)\tail(S+T)
 \approx\exp\!\left(\frac{(1+\theta)S^2-(S+T)^2}{2}\right).
\]
So we know two things about the distribution function
$F(s)=\Pr[M_{Y_p}(g)\leq s]$ of the inner maximum: $F(S)\leq1-\delta$ at
the base threshold, and $F(S+T)$ is extremely close to one at the farther
threshold.  Ehrhard's inequality says that $s\mapsto\Phi^{-1}(F(s))$ is
concave, because $\{M\leq s\}$ is a convex set that moves linearly in $s$.
A concave function lies below its secant lines.  The secant through the
points at $S$ and at $S+T$ therefore bounds the function at $S-a$ from
above, and this bound on $\Phi^{-1}(F(S-a))$ gives a bound on
$F(S-a)=\Pr[B_p]$.  Working out the numbers,
\[
 \Pr[M_{Y_p}(g)\leq S-a]
 \leq\exp\!\left(-\frac{1+\theta}{2\theta}a^2\right),
\]
once $T$ is chosen optimally as $T=\theta S$.  This is
Theorem~\ref{thm:two-threshold} together with
Lemma~\ref{lem:asymptotic-two-threshold}.  Remark~\ref{rem:calibration}
explains why $T=\theta S$ is the best choice.

Compare the two bounds.  Borell gives decay $e^{-a^2/2}$.  The two-threshold
inequality gives decay $e^{-(1+\theta)a^2/(2\theta)}$, which is faster
by the factor $(1+\theta)/\theta$ in the exponent.  To make
$\Pr[B_p]\leq\exp(-\Sigma t^2/2)$ it now suffices to take
\[
 a=\sqrt{\frac{\theta\Sigma}{1+\theta}}\,t
 \qquad\text{instead of}\qquad
 a=\sqrt\Sigma\,t.
\]
In our application $\theta\leq c\approx0.207$, so the new loss is at most
$\sqrt{c/(1+c)}\approx0.41$ times the old one.  This is the analytic source
of the improved exponent.

Two details matter in the application.  First, the inner vectors are not
unit vectors: they are projections of unit vectors orthogonal to a plane,
and have length $q<1$.  Normalizing them raises their threshold from
$\sqrt3t$ to $\sqrt3t/q$, which is where the quantity $\theta$ comes from.
Second, the inner family has at most $D\approx\exp(3(1+c)t^2/2)$ members
because the graph has maximum degree at most $D$.  Combining these,
\[
 \exp\!\left(\frac{(1+\theta)}{2}\cdot\frac{3t^2}{q^2}\right)
 =\exp\!\left(\frac{3(1+c)t^2}{2}\right)
 \quad\Longleftrightarrow\quad
 \theta=(1+c)q^2-1.
\]
So $\theta$ measures how much room is left after the projection: the
shorter the projected inner vectors, the smaller $\theta$, and the smaller
the loss.

\subsection{The recursion along walks}
\label{subsec:overview-recursion}

We now describe the geometry that the recursion tracks.  Fix a vertex $i$,
called the \emph{root}.  All walks start at $i$.  The final vertex of a walk
is its \emph{endpoint}, and the walk's \emph{level} is its number of edges.
For an endpoint $x$, its \emph{root correlation} and its normalized
projected direction are
\[
 h_x=\langle v_i,v_x\rangle,
 \qquad
 z_x=\frac{v_x-h_xv_i}{\sqrt{1-h_x^2}}
 \quad\text{when }|h_x|<1.
\]
The vector $z_x$ is the direction of $v_x$ orthogonal to the root.  Both
$h_x$ and $z_x$ depend only on the endpoint $x$, not on the walk that
reached it.  This will let us merge walks with the same endpoint.

Consider a two-step walk $i\to j\to k$.  Applying the edge identity twice,
\[
 v_k=\frac14v_i-\frac{\sqrt3}{4}v_{ij}+\frac{\sqrt3}{2}v_{jk},
 \qquad
 \langle v_i,v_k\rangle
 =\frac14+\frac34\langle v_{ji},v_{jk}\rangle.
\]
Among the walks retained after a failed rounding,
$\langle v_{ji},v_{jk}\rangle\geq-o(1)$, so second-level endpoints have
root correlation at least $1/4-o(1)$.  The output of the rounding says
that each retained first-edge direction $v_{ij}$ comes with a packing of
next-edge directions $v_{jk}$.  Combining the cover by the $v_{ij}$ with
these attached packings, and projecting orthogonally to the root, gives a
cover by the directions $z_k$ of the second-level endpoints.  Longer walks
repeat the same operation.  One small point: the coefficient of $v_{ij}$
above is negative, so the recursion uses the outer directions $-v_{ij}$;
this costs nothing because $g$ and $-g$ have the same distribution.

Lemma~\ref{lem:recursive-cover-step} packages one such step.  It remembers
only three summaries of the current endpoint set:
\begin{itemize}
\item $h$, the common root correlation of the endpoints, up to $o(1)$;
\item $\Lambda$, the threshold coefficient of the cover by the projected
      endpoint directions, so the threshold is $\Lambda t-o(t)$; and
\item $\mu$, the exponent in the upper bound $\exp(\mu t^2/2+o(t^2))$ on
      the number of distinct endpoints.
\end{itemize}
Each current endpoint $p$ also carries the packing of its next-edge
directions $v_{pm}$ at threshold $\sqrt3t$, and we write
$\langle v_i,v_{pm}\rangle=\eta+o(1)$ for the retained pairs.  Put
$r=\sqrt{1-h^2}$ and $\tau=\sqrt3/2$.  The three quantities from the
previous subsection become
\[
 \Sigma=\mu-\Lambda^2,
 \qquad
 \theta=c-(1+c)\frac{\eta^2}{r^2},
 \qquad
 \chi=\sqrt{\frac{\theta\Sigma}{1+c}}.
\]
Here $\Sigma$ is the surplus of the outer cover.  Projecting $v_{pm}$
orthogonally to the plane spanned by $v_i$ and $v_p$ leaves a vector of
squared length $q^2=1-\eta^2/r^2$, and $\theta=(1+c)q^2-1$ is the room
left after this projection, as computed above.  Finally $\chi t$ is the
loss $a$ from the previous subsection, converted back to the unnormalized
inner vectors by the factor $q$:
\[
 q\sqrt{\frac{\theta\Sigma}{1+\theta}}
 =\sqrt{\frac{1+\theta}{1+c}\cdot\frac{\theta\Sigma}{1+\theta}}
 =\chi.
\]
The geometry of the next endpoint $m$ is given by
\[
 C=\frac r2+\tau\frac{h\eta}{r},
 \qquad
 h'=-\frac h2+\tau\eta,
\]
where $h'$ is the root correlation of $m$, and $C$ is the coefficient with
which the old projected direction $z_p$ survives in the new one.  The new
endpoints then form a cover at threshold $\Lambda't-o(t)$, where
\begin{equation}
 \Lambda'
 =\frac{\Lambda C+\frac32-\tau\chi}
        {\sqrt{1-(h')^2}}.
 \label{eq:overview-recursion}
\end{equation}
Here is how to read this formula.  The term $\Lambda C$ is the part of the
old threshold that survives the change of direction, and
$3/2=\tau\sqrt3$ comes from the attached packing at threshold $\sqrt3t$.
Requiring both crossings to use the same walk costs $\tau\chi$.  The
denominator rescales the new projected direction to unit length.

The baseline $\Lambda^2$ in $\Sigma=\mu-\Lambda^2$ is the union-bound
count from the previous subsection: a cover at threshold $\Lambda t$ with
subexponential probability has at least $\exp(\Lambda^2t^2/2-o(t^2))$
members.  In our application the members are distinct endpoints.  The same
observation will eventually rule out the fourth-level cover.

\subsection{Buckets}
\label{subsec:overview-buckets}

One nuisance remains: correlations such as $h$ and $\eta$ are not exactly
the same for every walk.  We therefore divide $[-1,1]$ into
$Q=\lceil\log t\rceil$ intervals of equal length, called \emph{buckets},
and group walks whose correlations fall in the same buckets.  Within one
group, a single representative value describes every walk up to $o(1)$.
There are only $Q^{O(1)}$ groups, so keeping the heaviest group loses a
factor $Q^{O(1)}$ in packing mass, which keeps the mass subexponential.
Lemma~\ref{lem:common-bucket} formalizes this selection.  In the algorithm,
the buckets are what is enumerated: for every root and every choice of one
bucket per relevant correlation, the algorithm lists all endpoints of short
walks with correlations in those buckets.

\subsection{Round or continue}
\label{subsec:overview-round-or-continue}

Every intermediate endpoint family gives one of two useful outcomes: either
it can be rounded to a large independent set, or its size is small enough to
continue the recursion.  The only subtlety is that the proof works with a
hidden family of endpoints, whereas the algorithm can construct only a
larger one.

The \emph{good endpoints} are the distinct endpoints whose projected
directions carry the covering probability established by the recursion.
They are hidden because the proof selects them through packing weights and
bucket choices that the algorithm does not know.  Instead, for every root
and every choice of buckets, the algorithm lists all distinct endpoints of
walks of the required length whose correlations lie in the chosen buckets.
We call this list the \emph{candidate set}.  For the buckets selected in
the proof, the candidate set contains the good endpoints.

At levels two and three, the proof fixes a cutoff on the size of the
candidate set, of the form $\exp(\mu t^2/2)$ for a coefficient $\mu$ that
depends on the level and on the buckets.  The two cases are:
\begin{itemize}
\item If the candidate set is at least as large as the cutoff, then the
      root correlations in its buckets give a good vector coloring of the
      set, and KMS rounding finds the desired independent set.
\item If the candidate set is smaller than the cutoff, then so is the set
      of good endpoints.  This supplies the coefficient $\mu$ needed to
      extend the cover by one more edge.
\end{itemize}
This is Lemma~\ref{lem:cardinality-split}.  The case split happens only in
the proof.  The algorithm runs every candidate rounding, so it never needs
to know which case occurred.

For the rounding step, recall that a vector $\kappa$-coloring has edge
inner products at most $-1/(\kappa-1)$, and that KMS rounding on a set of
maximum degree $D$ loses a factor $D^{1-2/\kappa}$ up to subpolynomial
terms.  At level two, the endpoints have root correlation about $1/4$ or
more, and Lemma~\ref{lem:bhl-positive-coloring} gives a vector coloring
with $\kappa\approx2$, so almost nothing is lost.  At level three, the
available coloring depends on the root correlation $h$ of the endpoints in
the selected buckets: for $h\leq0$, Lemma~\ref{lem:bhl-negative-coloring}
gives $\kappa=(3-6h)/(1-4h)$, and for $h\geq0$ we fall back on the strict
vector $3$-coloring.

The BHL analysis uses the same two-case step through the third level.  Our
smaller loss leaves enough room to extend once more when both intermediate
sets are below their cutoffs.  The fourth-level cover is used only for
counting: by the union-bound count above, it needs at least
$\exp(\Lambda_4^2t^2/2-o(t^2))$ distinct endpoints, and the numbers work
out so that this exceeds $n$.  The whole bounded-degree argument is
summarized below; ``large'' and ``small'' mean relative to the cutoff.
\begin{center}
\small
\renewcommand{\arraystretch}{1.08}
\begin{tabular}{c|>{\raggedright\arraybackslash}p{0.36\linewidth}|>{\raggedright\arraybackslash}p{0.46\linewidth}}
Level & What is available & What the argument does \\
\hline
1 & The weighted first-step directions supplied by the failed rounding
  & The maximum-degree bound limits the number of neighbors; extend the
    cover to endpoints of two-edge walks \\
2 & A cover on two-edge endpoints $x$ with
    $\langle v_i,v_x\rangle\geq1/4-o(1)$
  & If the candidate set is large, use the coloring for positive root
    correlation and round; if it is small, use the size bound and extend
    to level three \\
3 & A cover on three-edge endpoints $x$ whose root correlation lies in one
    bucket
  & If the candidate set is large, use the coloring for nonpositive root
    correlation or the strict vector $3$-coloring and round; if it is
    small, extend once more \\
4 & A cover at the final recursively computed threshold
  & A union bound forces more than $n$ distinct endpoints, a contradiction
\end{tabular}
\end{center}

\paragraph{A fixed margin.}
Every numerical inequality in this argument holds with room to spare.  We
use this room in a simple way.  Fix once and for all a small constant
$\omega_0>0$.  Each cutoff is padded to $\exp((\mu+\omega_0)t^2/2)$, and in
the recursion step the surplus $\Sigma$ and the room $\theta$ are both
replaced by $\Sigma+\omega_0$ and $\theta+\omega_0$.  The padding makes
every quantity that must be positive at least $\omega_0$, so all the
lower-order errors, which tend to zero, are eventually smaller than the
padding.  Because the coefficients depend continuously on $\omega_0$ and
the numerical inequalities are strict at $\omega_0=0$, a sufficiently small
$\omega_0$ preserves them.  This is Lemma~\ref{lem:margin}.  The padding
also handles the boundary cases $\Sigma=0$ and $\theta=0$ without separate
treatment.

\subsection{The algorithm and the exponent}
\label{subsec:overview-algorithm}

Although the good endpoints are hidden, every candidate set is explicit.
The algorithm solves the three-round relaxation, repeats the modified KMS
rounding on the whole graph, and then tries every root, every choice of one
bucket at level two, and every ordered pair of buckets at level three.  For
each choice, it enumerates the relevant walks of length at most three,
merges repeated endpoints, constructs the appropriate vector coloring, and
runs KMS rounding on the candidate set.  There are
$1+nQ+nQ^2$ candidate families in all, which is polynomial in $n$.  The
fourth level is never enumerated; it is only the contradiction proving that
one of the earlier calls succeeds.

The numerical verification is also narrow.  The side conditions for the
first transition are verified by hand.  For the two later,
parameter-dependent transitions, the certificate in
Appendix~\ref{sec:numerics} checks by interval arithmetic that every
square root and normalization in the recursion is well defined, that the
level-two rounding gives enough, and that the fourth-level cover is too
large.  It is an offline verification of fixed constants, not a step of the
algorithm.

Finally, we balance the dense and sparse regimes.  Take
$\widehat c=0.2067$ and set
\[
 \widehat\alpha=\frac1{5+3\widehat c},
 \qquad
 \widehat\beta=\frac{3+3\widehat c}{5+3\widehat c},
 \qquad
 D=n^{\widehat\beta}.
\]
Graphs of minimum degree at least $D$ are handled by the dense-graph
routine of \citet*{KTY24}, which makes progress toward an
$n^{(1-\widehat\beta)/2+o(1)}=n^{\widehat\alpha+o(1)}$-coloring.  On
graphs of maximum degree at most $D$, we choose the Gaussian threshold with
$\tail(t)=D^{-1/[3(1+\widehat c)]}=n^{-\widehat\alpha}$.  A successful
rounding gives an independent set of size about $n\tail(t)=n^{1-\widehat\alpha}$,
which is what progress toward an $n^{\widehat\alpha}$-coloring requires.  A
failed rounding gives, by Proposition~\ref{prop:failed-rounding-progress},
an independent set of size $D^{A_0-o(1)}$ with $A_0=3191/2500$.  Writing
$R(x)=(4+3x)/(3+3x)$, the parameters satisfy $D^{R(\widehat c)}=n^{1-\widehat\alpha}$,
and the certificate verifies $A_0>R(\widehat c)$.  So the failed-rounding
branch also reaches the required size.  Since
$\widehat\alpha=0.1779327\ldots<0.17794$, the dense--sparse combination
theorem yields Theorem~\ref{thm:final-coloring}.  The recursion itself is
run with a slightly larger constant $c=0.207$; the gap between $c$ and
$\widehat c$ is what makes the degree bound $D$ strictly smaller than the
cardinality bound $\exp(3(1+c)t^2/2)$ that the recursion assumes.

\paragraph{Division of labor.}
Earlier work supplies the strict vector coloring, KMS rounding, the packing
theorems that prune a failed rounding, the conditional vector colorings,
and the dense--sparse framework.  The new work in this paper has four
parts:
\begin{itemize}
\item a Gaussian comparison that remembers both the number and the lengths
      of the inner vectors;
\item a maximum-degree version of the pruning step, suitable for repeated
      composition;
\item the single recurrence~\eqref{eq:overview-recursion} organizing the
      successive walk levels; and
\item the final round-or-continue step whose fourth-level cover contradicts
      the number of graph vertices.
\end{itemize}
The rest of the paper states each step formally and verifies every
geometric and numerical condition.

\section{Preliminaries}
\label{sec:preliminaries}

This section begins the formal part of the paper.  It fixes notation,
records the standard Gaussian estimates we use, defines covers, packings,
and buckets, and states the results quoted from earlier work.

\subsection{Notation}

All graphs are finite, simple, and undirected, and all logarithms are
natural.  Let $G=(V,E)$, and write $n=|V|$.  We use $N_G(v)$ and
$\deg_G(v)$ for the neighborhood and degree of $v$, and $\Delta(G)$ for the
maximum degree; we omit the subscript when the graph is clear.  For a
vertex set $S$, let $N_G(S)$ be the set of vertices outside $S$ with a
neighbor in $S$, and let $G[S]$ be the induced subgraph.  A set of vertices
is independent if it contains no edge.  We write $\langle x,y\rangle$ and
$\|x\|$ for the Euclidean inner product and norm, and $x_+=\max\{x,0\}$.

Let $g\sim N(0,I_d)$ be a standard Gaussian vector in $\mathbb R^d$, and
let $\gamma_d$ denote the standard Gaussian measure.  We write
$\phi(x)=(2\pi)^{-1/2}e^{-x^2/2}$ for the standard normal density,
$\Phi$ for its distribution function, $\tail=1-\Phi$ for the tail, and
$\tail^{-1}$ for the inverse of the tail, so that
$\tail^{-1}(p)=\Phi^{-1}(1-p)$.

Throughout the sparse-graph analysis, a Gaussian threshold $t$ tends to
infinity, and every asymptotic statement refers to $t\to\infty$ unless
stated otherwise.

\subsection{Gaussian estimates}

\begin{lemma}[Gaussian tail and quantile estimates]
\label{lem:gaussian-estimates}
For every $x>0$,
\begin{equation}
 \frac{x}{1+x^2}\phi(x)
 \leq\tail(x)
 \leq\frac{\phi(x)}x,
 \label{eq:gaussian-mills}
\end{equation}
and consequently
\begin{equation}
 \log\tail(x)=-\frac{x^2}{2}+O(\log(x+2)).
 \label{eq:tail-log-uniform}
\end{equation}
Moreover,
\begin{align}
 \tail(x)&\leq e^{-x^2/2}
 \qquad(x\geq0),
 \label{eq:gaussian-chernoff}\\
 \bigl(\tail^{-1}(p)\bigr)_+&\leq\sqrt{2\log(1/p)}
 \qquad(0<p<1),
 \label{eq:quantile-upper}
\end{align}
and there is a universal constant $C_0$ such that, for every sufficiently
large $H$,
\begin{equation}
 \tail^{-1}(e^{-H})
 \geq\sqrt{2H}-C_0\frac{\log(H+2)}{\sqrt H}.
 \label{eq:quantile-lower}
\end{equation}
\end{lemma}

\begin{proof}
The bounds~\eqref{eq:gaussian-mills} are the classical Mills-ratio bounds;
see, for example, Lemma~A.1 in the full version of \citet*{CMM06}.  Taking
logarithms in~\eqref{eq:gaussian-mills} for $x\geq1$, and using the
boundedness of $\log\tail(x)+x^2/2$ on $(0,1]$, gives
\eqref{eq:tail-log-uniform}.

For~\eqref{eq:gaussian-chernoff}, let $Z\sim N(0,1)$.  For $x>0$, Markov's
inequality applied to $e^{xZ}$ and the identity
$\mathbb Ee^{xZ}=e^{x^2/2}$ give
$\tail(x)=\Pr[e^{xZ}\geq e^{x^2}]\leq e^{-x^2}\mathbb Ee^{xZ}=e^{-x^2/2}$.
The case $x=0$ is trivial.  For~\eqref{eq:quantile-upper}, if
$x=\tail^{-1}(p)>0$, then $p=\tail(x)\leq e^{-x^2/2}$ by
\eqref{eq:gaussian-chernoff}, which is the claim.

Finally, let $H$ be large and put $x=\tail^{-1}(e^{-H})>0$.
Equation~\eqref{eq:tail-log-uniform} gives $H=x^2/2+O(\log(x+2))$, and
\eqref{eq:quantile-upper} gives $x\leq\sqrt{2H}$.  Hence
$x^2\geq2H-C\log(H+2)$ for a universal constant $C$, and taking square
roots proves~\eqref{eq:quantile-lower}.
\end{proof}

By symmetry, \eqref{eq:gaussian-chernoff} also reads
$\Phi(-y)\leq e^{-y^2/2}$ for $y\geq0$, and hence
\begin{equation}
 \log\Phi(x)\leq-\frac12\bigl(x_-\bigr)^2
 \qquad(x\in\mathbb R),
 \qquad\text{where }x_-=\max\{-x,0\}.
 \label{eq:log-phi-bound}
\end{equation}

\subsection{Vector colorings}

For $\kappa\geq2$, a \emph{vector $\kappa$-coloring} of a graph assigns a
unit vector $v_i$ to each vertex $i$ and requires
\[
 \langle v_i,v_j\rangle\leq-\frac1{\kappa-1}
 \quad\text{for every edge }\{i,j\}.
\]
It is \emph{strict} when equality holds on every edge.  Following
\citet*{BHL26}, for any pair of nonparallel unit vectors $v_i,v_j$, define
the unit vector
\[
 v_{ij}
 =
 \frac{v_j-\langle v_i,v_j\rangle v_i}
      {\sqrt{1-\langle v_i,v_j\rangle^2}}
 \perp v_i .
\]
In a strict vector $3$-coloring, every edge $\{i,j\}$ satisfies
$\langle v_i,v_j\rangle=-1/2$, and hence
\begin{equation}
 v_j=-\frac12v_i+\tau v_{ij},
 \qquad
 \tau=\frac{\sqrt3}{2}.
 \label{eq:edge-identity}
\end{equation}
We use the constant $\tau=\sqrt3/2$ throughout.

\subsection{Covers, packings, and subexponential probabilities}

\begin{definition}[Gaussian cover]
Let $X=\{x_j:j\in J\}$ be a finite family of unit vectors.  For a threshold
$s\in\mathbb R$ and a probability $\delta\in(0,1]$, we call $X$ an
$(s,\delta)$-\emph{cover} if
\[
 \Pr\!\left[\max_{j\in J}\langle g,x_j\rangle\geq s\right]\geq\delta.
\]
\end{definition}

An ordinary cover specifies the probability that at least one vector
crosses the threshold, but it does not say how much of this probability is
carried by a chosen subfamily.  The following weighted version, taken from
\citet*{BHL26}, assigns mass to individual vectors in a way that survives
restriction.

\begin{definition}[Packing measure]
Let $X=\{x_j:j\in J\}$ be a finite family of unit vectors, and let $\nu$ be
a nonnegative measure on $J$.  We call $\nu$ an $(s,\delta)$-\emph{packing
measure} for $X$ if
\begin{equation}
\begin{aligned}
 \nu(A)&\leq
 \Pr[\exists j\in A:\langle g,x_j\rangle\geq s]
 &&\text{for every }A\subseteq J,\\
 \nu(J)&\geq\delta.
\end{aligned}
\label{eq:packing-measure}
\end{equation}
We also call the weighted family $(X,\nu)$ an $(s,\delta)$-\emph{packing},
and $\nu(J)$ its \emph{mass}.
\end{definition}

Every $(s,\delta)$-cover admits a packing measure of mass equal to its
covering probability: fix an ordering of $J$ and assign each sample $g$ in
the cover event to the first index whose vector crosses the threshold.  This
is also Lemma~4.6 of \citet*{BHL26}.  Conversely, the first condition in
\eqref{eq:packing-measure} shows that any restriction of a packing to a
subfamily of mass $p$ is an $(s,p)$-cover.

Two conventions will be used without further comment.  First, results
stated for sets of vectors are also applied to indexed families.  If
several indices give the same vector, we combine their weights, apply the
result to the distinct vectors, and, when a vector is retained, retain all
of its indices.  This preserves the total mass and every inequality of the
form~\eqref{eq:packing-measure}.  Second, since $g$ and $-g$ have the same
distribution, replacing every vector $x_j$ by $-x_j$ changes neither the
covering probability nor the packing property.

\begin{definition}[Subexponential probabilities]
A sequence $\delta_t\in(0,1]$ is \emph{subexponential} if
$\log(1/\delta_t)=o(t^2)$, or equivalently $\delta_t=\exp[-o(t^2)]$.  A
cover or packing at threshold $s_t$ is \emph{subexponential} if its
probability or mass is a subexponential sequence.
\end{definition}

Thus a subexponential probability may tend to zero, but not exponentially
fast on the $t^2$ scale.  This is the relevant scale because the Gaussian
tail probabilities in the proof have the form $\exp[-\Theta(t^2)]$.
Multiplying a subexponential sequence by $(\log t)^{-O(1)}$ or by a fixed
positive constant preserves subexponentiality.

\subsection{Buckets}
\label{subsec:buckets}

Put $Q=\lceil\log t\rceil$ and fix one partition $\mathcal P_t$ of $[-1,1]$
into $Q$ intervals of equal length, called \emph{buckets}.  The buckets
are half-open, except that the last one contains $1$.  Every bucket has
width $2/Q=O(1/\log t)$.  The same partition is used in the proof and in
the algorithm, so the algorithm tries every bucket that the proof selects.

At several points, we have an outer packing together with an inner packing
attached to every outer index, and we need all the correlations describing
these families to lie in common buckets.  The following selection lemma
shows that this costs only a polylogarithmic factor in mass.

\begin{lemma}[Common-bucket selection]
\label{lem:common-bucket}
Let an outer packing have mass $\delta_0$, and attach to every outer index
an inner packing of mass at least $\delta_1$.  Suppose that each outer
index is labeled by one of $Q_{\rm out}$ outer labels, and each inner index
by one of $Q_{\rm in}$ inner labels.  Then there are an outer label and an
inner label such that the outer indices carrying the outer label have
outer mass at least $\delta_0/(Q_{\rm out}Q_{\rm in})$, and every one of
these outer indices has inner mass at least $\delta_1/Q_{\rm in}$ on the
inner indices carrying the inner label.
\end{lemma}

\begin{proof}
For each outer index, choose an inner label of maximum inner mass; it has
mass at least $\delta_1/Q_{\rm in}$.  Tag the outer index by the pair
consisting of its own label and this chosen inner label.  There are at most
$Q_{\rm out}Q_{\rm in}$ tags, so some tag carries outer mass at least
$\delta_0/(Q_{\rm out}Q_{\rm in})$.  Equation~\eqref{eq:packing-measure}
turns both restrictions into covers.
\end{proof}

In every application, a label is a tuple of at most three buckets, so
$Q_{\rm out},Q_{\rm in}=Q^{O(1)}$, and the loss is $Q^{-O(1)}$, which
preserves subexponential mass.  All bucket choices and tie-breaking rules
in the proof are fixed before the Gaussian vector of the next step is
sampled, so the families produced are deterministic.

\subsection{Ehrhard's inequality}

The following theorem is due to \citet*{Ehrhard83}; see also
Section~2.2 of \citet*{AW09} for a textbook treatment.

\begin{mainresultbox}
\begin{theorem}[Ehrhard's inequality]
\label{thm:ehrhard}
Let $A$ and $B$ be convex Borel subsets of $\mathbb R^d$, and let
$0\leq\lambda\leq1$.  Then
\[
 \Phi^{-1}\!\left(\gamma_d((1-\lambda)A+\lambda B)\right)
 \geq (1-\lambda)\Phi^{-1}(\gamma_d(A))
      +\lambda\Phi^{-1}(\gamma_d(B)).
\]
Here $(1-\lambda)A+\lambda B$ denotes the Minkowski combination.  We set
$\Phi^{-1}(0)=-\infty$ and $\Phi^{-1}(1)=+\infty$; if one of these values
occurs, the inequality is understood by taking limits from Gaussian
measures strictly between zero and one.
\end{theorem}
\end{mainresultbox}

\subsection{Results quoted from earlier work}
\label{subsec:quoted}

The vertex vectors $v_i$ obtained from the three-round relaxation of
Section~\ref{subsec:overview-sdp} form a strict vector $3$-coloring.  As
usual in SDP analyses, we treat the solution as exact.  We use the following
two conditional vector-coloring results of \citet*[Lemmas~3.2 and~3.4]{BHL26}.
We state their restrictions to an arbitrary vertex set $U$, obtained by
applying the published statements to the subgraph induced by
$U\cup\{i\}$.

\begin{lemma}[Coloring for nonpositive root correlation \citep{BHL26}]
\label{lem:bhl-negative-coloring}
For the symmetrized three-round solution, fix a vertex $i$ and a set
$U\subseteq V\setminus\{i\}$.  Suppose that
\[
 -\frac12\leq h\leq0,
 \qquad
 \langle v_i,v_j\rangle\leq h
 \quad\text{for every }j\in U.
\]
Then the vectors
\[
 b_j=
 y_{\{(i,\mathsf R),(j,\mathsf G)\}}
 -y_{\{(i,\mathsf R),(j,\mathsf B)\}},
 \qquad j\in U,
\]
are nonzero, and the normalized vectors $b_j/\|b_j\|$ form a vector
$\dfrac{3-6h}{1-4h}$-coloring of $G[U]$.  This vector coloring can be
constructed in polynomial time from the three-round relaxation.
\end{lemma}

\begin{lemma}[Coloring for positive root correlation \citep{BHL26}]
\label{lem:bhl-positive-coloring}
For the symmetrized three-round solution, fix a vertex $i$ and a set
$U\subseteq V\setminus\{i\}$.  Suppose that
\[
 \frac1{16}\leq\vartheta\leq\frac14,
 \qquad
 \langle v_i,v_j\rangle\geq\vartheta
 \quad\text{for every }j\in U.
\]
Then the vectors
\[
 a_j=
 y_{\{(i,\mathsf R),(j,\mathsf R)\}}
 -y_{\{(i,\mathsf R),(j,\mathsf G)\}}
 -y_{\{(i,\mathsf R),(j,\mathsf B)\}},
 \qquad j\in U,
\]
are nonzero, and the normalized vectors $a_j/\|a_j\|$ form a vector
$\dfrac{4+8\vartheta}{1+8\vartheta}$-coloring of $G[U]$.  This vector
coloring can be constructed in polynomial time from the three-round
relaxation.
\end{lemma}

The following rounding theorem is due to \citet*{KMS98};
\citet*[Theorem~2.3]{BHL26} state the real-parameter form.  We need the
bound to be uniform in $\kappa$ near $2$, because at level two we round a
vector $\kappa$-coloring with $\kappa=2+o(1)$.  For completeness, the
short proof is given in Appendix~\ref{sec:kms-proof}.

\begin{theorem}[KMS rounding \citep{KMS98,BHL26}]
\label{thm:kms-extraction}
Fix $K<\infty$.  Uniformly for $2\leq\kappa\leq K$, given a vector
$\kappa$-coloring of a graph $H$ with $m$ vertices and maximum degree at
most $D\geq1$, randomized polynomial-time KMS rounding finds an independent
set of size
\[
 \Omega_K\!\left(
   \frac{mD^{-(1-2/\kappa)}}{\sqrt{\log(2D)}}
 \right).
\]
\end{theorem}

The other results quoted from earlier work are the two packing theorems of
\citet*{BHL26}, stated in Section~\ref{sec:structural-input} where they are
used, and the two dense--sparse results stated in
Section~\ref{sec:coloring-proof}.

\section{The two-threshold inequality}
\label{sec:two-threshold}

This section proves the analytic ingredient used in every step of the
recursion.  Section~\ref{subsec:overview-two-threshold} explains the idea:
Ehrhard's inequality makes the distribution function of a Gaussian maximum
concave after the change of variable $\Phi^{-1}$, and a concave function
lies below its secant lines.  Knowing the distribution function at a base
threshold and at a farther threshold therefore bounds it at a lower
threshold.  The estimate at the farther threshold comes from a union bound,
so it remembers how many vectors the family contains.

Throughout this section, $X=\{x_1,\ldots,x_m\}$ is a finite family of unit
vectors in $\mathbb R^d$, and
\[
 M(x)=\max_{1\leq j\leq m}\langle x,x_j\rangle,
 \qquad
 K_s=\{x\in\mathbb R^d:M(x)\leq s\},
 \qquad
 F(s)=\Pr[M(g)\leq s]=\gamma_d(K_s).
\]
The set $K_s$ is an intersection of $m$ halfspaces, one for each vector,
and increasing $s$ moves every halfspace outward.  Since each
$\langle g,x_j\rangle$ is a standard normal variable,
$\Pr[M(g)=s]\leq\sum_j\Pr[\langle g,x_j\rangle=s]=0$ for every $s$, so
$F(s)=\Pr[M(g)<s]$ as well.

\begin{lemma}[Concavity]
\label{lem:concavity}
The function $s\mapsto\Phi^{-1}(F(s))$ is concave on $\mathbb R$, with the
conventions $\Phi^{-1}(0)=-\infty$ and $\Phi^{-1}(1)=+\infty$ of
Theorem~\ref{thm:ehrhard}.
\end{lemma}

\begin{proof}
The function $M$ is convex, so for $0\leq\lambda\leq1$ and
$x\in K_s$, $y\in K_{s'}$ we have
$M(\lambda x+(1-\lambda)y)\leq\lambda s+(1-\lambda)s'$.  Thus
\[
 \lambda K_s+(1-\lambda)K_{s'}
 \subseteq K_{\lambda s+(1-\lambda)s'}.
\]
Each $K_s$ is closed and convex.  Theorem~\ref{thm:ehrhard}, followed by
the monotonicity of $\Phi^{-1}$, gives
$\Phi^{-1}(F(\lambda s+(1-\lambda)s'))
 \geq\lambda\Phi^{-1}(F(s))+(1-\lambda)\Phi^{-1}(F(s'))$.
\end{proof}

This concavity of Gaussian suprema is classical; see \citet*{Bobkov08} for
a discussion of what it says about the possible distributions of a Gaussian
maximum.

\begin{mainresultbox}
\begin{theorem}[Two-threshold inequality]
\label{thm:two-threshold}
Let $X=\{x_1,\ldots,x_m\}$ be unit vectors, let $S\in\mathbb R$ be a base
threshold, and let $T>0$ be a remote displacement.  Suppose that, for some
$\delta,\varepsilon\in(0,1)$,
\[
 \Pr[M(g)\geq S]\geq\delta,
 \qquad
 \Pr[M(g)>S+T]\leq\varepsilon.
\]
Then, for every $a\geq0$,
\begin{equation}
 \Pr[M(g)\leq S-a]
 \leq
 \Phi\!\left[
  \left(1+\frac aT\right)\tail^{-1}(\delta)
  -\frac aT\tail^{-1}(\varepsilon)
 \right].
 \label{eq:exact-two-threshold}
\end{equation}
Moreover, the union bound gives
\begin{equation}
 \Pr[M(g)>S+T]\leq m\,\tail(S+T),
 \label{eq:union-anchor}
\end{equation}
so \eqref{eq:exact-two-threshold} holds with
$\varepsilon=m\tail(S+T)$ whenever this number is less than one.
\end{theorem}
\end{mainresultbox}

\begin{proof}
Write $q(s)=\Phi^{-1}(F(s))$.  Fix $a\geq0$.  If $F(S-a)=0$, the claim is
trivial, so assume $F(S-a)>0$.  The base hypothesis gives
$F(S)=1-\Pr[M(g)\geq S]\leq1-\delta<1$, and hence
$0<F(S-a)\leq F(S)<1$.  The remote hypothesis gives
$F(S+T)\geq1-\varepsilon>0$, and $F(S+T)<1$ because $K_{S+T}$ is contained
in the halfspace $\{x:\langle x,x_1\rangle\leq S+T\}$, whose Gaussian
measure is less than one.  Thus $q(S-a)$, $q(S)$, and $q(S+T)$ are finite.

The base threshold is a convex combination of the other two thresholds:
\[
 S=\frac{T}{T+a}(S-a)+\frac{a}{T+a}(S+T).
\]
Lemma~\ref{lem:concavity} gives
$q(S)\geq\frac{T}{T+a}q(S-a)+\frac{a}{T+a}q(S+T)$, and solving for
$q(S-a)$,
\[
 q(S-a)\leq\left(1+\frac aT\right)q(S)-\frac aTq(S+T).
\]
The two hypotheses say
$q(S)\leq\Phi^{-1}(1-\delta)=\tail^{-1}(\delta)$ and
$q(S+T)\geq\Phi^{-1}(1-\varepsilon)=\tail^{-1}(\varepsilon)$.  Substituting
and applying the increasing function $\Phi$ to both sides proves
\eqref{eq:exact-two-threshold}.  Finally, $\{M(g)>S+T\}$ is the union of
the events $\{\langle g,x_j\rangle>S+T\}$, each of probability
$\tail(S+T)$, which gives~\eqref{eq:union-anchor}.
\end{proof}

The next lemma is the asymptotic form used in the recursion.  It applies
to a family whose threshold is of order $t$, whose cardinality is at most
$\exp((1+\theta)S^2/2)$ up to lower-order terms, and whose covering
probability is subexponential.  The remote displacement is chosen as
$T=\theta S$; Remark~\ref{rem:calibration} explains why.  The lemma allows
$\theta$ to vary with $t$ within a fixed compact interval, and the
remainder is uniform over all families satisfying the hypotheses with the
same bounds.  This uniformity is what lets us apply the lemma
simultaneously to the many inner families of one recursion step.

\begin{lemma}[Asymptotic two-threshold comparison]
\label{lem:asymptotic-two-threshold}
Fix constants $0<\theta_0\leq M<\infty$, $A>0$, and
$0<s_-\leq s_+<\infty$, and let $t\to\infty$.  For each $t$, let $X_t$ be
a family of $m_t$ unit vectors, in any dimension, that is an
$(S_t,\delta_t)$-cover with $s_-t\leq S_t\leq s_+t$, and let
$\theta_t\in[\theta_0,M]$.  Suppose that
\begin{equation}
 \log m_t\leq\frac{1+\theta_t}{2}S_t^2+R_t,
 \qquad
 \log\frac1{\delta_t}\leq L_t,
 \qquad
 R_t,L_t=o(t^2).
 \label{eq:asymptotic-hypotheses}
\end{equation}
Then, uniformly over $a\in[0,At]$,
\begin{equation}
 \log\Pr\!\left[\max_{x\in X_t}\langle g,x\rangle\leq S_t-a\right]
 \leq-\frac{1+\theta_t}{2\theta_t}a^2+o(t^2),
 \label{eq:asymptotic-conclusion}
\end{equation}
with the convention $\log0=-\infty$.  The $o(t^2)$ term depends only on
the constants $\theta_0,M,A,s_-,s_+$ and on the sequences $R_t$ and $L_t$.
In particular, if several families satisfy
\eqref{eq:asymptotic-hypotheses} with the same $R_t$ and $L_t$, then
\eqref{eq:asymptotic-conclusion} holds for all of them with a common
remainder.
\end{lemma}

\begin{proof}
Apply Theorem~\ref{thm:two-threshold} to $X_t$ with base threshold $S_t$
and remote displacement $T_t=\theta_tS_t$.  The cover hypothesis gives
$\Pr[M(g)\geq S_t]\geq\delta_t$.  For the remote estimate, put
$\varepsilon_t=m_t\tail((1+\theta_t)S_t)$.  By~\eqref{eq:union-anchor},
$\Pr[M(g)>S_t+T_t]\leq\varepsilon_t$, and by
\eqref{eq:gaussian-chernoff} and the cardinality hypothesis,
\[
 \log\varepsilon_t
 \leq\frac{1+\theta_t}{2}S_t^2+R_t-\frac{(1+\theta_t)^2}{2}S_t^2
 =-H_t,
 \qquad
 H_t:=\frac{\theta_t(1+\theta_t)}{2}S_t^2-R_t.
\]
Since $H_t\geq\theta_0s_-^2t^2/2-R_t\to\infty$, we have $\varepsilon_t<1$
for all large $t$, and the theorem applies.

We translate the two probabilities into quantiles.  Put
$B_t=\sqrt{\theta_t(1+\theta_t)}\,S_t$, so that $2H_t=B_t^2-2R_t$ and
\[
 \theta_0s_-^2t^2\leq B_t^2\leq M(1+M)s_+^2t^2.
\]
For all large $t$, $2R_t\leq B_t^2/2$, and then $H_t\geq B_t^2/4$ and
$\sqrt{2H_t}=B_t\sqrt{1-2R_t/B_t^2}\geq B_t-2R_t/B_t$.  Since
$\tail^{-1}$ is decreasing, \eqref{eq:quantile-lower} gives
\[
 \tail^{-1}(\varepsilon_t)
 \geq\tail^{-1}(e^{-H_t})
 \geq\sqrt{2H_t}-C_0\frac{\log(H_t+2)}{\sqrt{H_t}}
 \geq B_t-\eta_t,
 \qquad
 \eta_t:=\frac{2R_t}{\sqrt{\theta_0}\,s_-t}
        +\frac{2C_0\log(B_t^2+2)}{B_t}.
\]
The displayed bounds on $B_t$ show that $\eta_t=o(t)$ and that $\eta_t$
depends only on the constants and on $R_t$.  At the base threshold,
\eqref{eq:quantile-upper} gives
$\tail^{-1}(\delta_t)\leq b_t:=\sqrt{2L_t}=o(t)$.

Now fix $a\in[0,At]$ and put
$p_a=\Pr[M(g)\leq S_t-a]$.  Inequality~\eqref{eq:exact-two-threshold}
gives
\[
 \Phi^{-1}(p_a)
 \leq\left(1+\frac{a}{\theta_tS_t}\right)b_t
     -\frac{a}{\theta_tS_t}(B_t-\eta_t)
 =-\sqrt{\frac{1+\theta_t}{\theta_t}}\,a
  +\left(1+\frac{a}{\theta_tS_t}\right)b_t
  +\frac{a}{\theta_tS_t}\eta_t,
\]
where we used $aB_t/(\theta_tS_t)=a\sqrt{(1+\theta_t)/\theta_t}$.  Since
$a/(\theta_tS_t)\leq A':=A/(\theta_0s_-)$, the last two terms are at most
\[
 r_t:=(1+A')b_t+A'\eta_t=o(t),
\]
which does not depend on $a$.  Put $\kappa_t=\sqrt{(1+\theta_t)/\theta_t}
\leq\sqrt{(1+\theta_0)/\theta_0}$.  The monotonicity of $\Phi$ and
\eqref{eq:log-phi-bound} give
\[
 \log p_a
 \leq\log\Phi(-\kappa_ta+r_t)
 \leq-\frac12(\kappa_ta-r_t)_+^2
 \leq-\frac{\kappa_t^2}{2}a^2+\kappa_tar_t
 \leq-\frac{1+\theta_t}{2\theta_t}a^2
     +\sqrt{\frac{1+\theta_0}{\theta_0}}\,Atr_t,
\]
where the third inequality uses $(x-y)_+^2\geq x^2-2xy$ for $x,y\geq0$.
The last term is $o(t^2)$, uniformly in $a$, and depends only on the
constants and on $R_t$ and $L_t$.
\end{proof}

\begin{remark}[Choice of the remote displacement]
\label{rem:calibration}
Suppose that the family has $\exp((1+\theta)S^2/2)$ members and consider a
remote threshold $xS$ with $x>\sqrt{1+\theta}$.  The union bound gives
\[
 -\log\Pr[M(g)>xS]
 \geq\frac{x^2-(1+\theta)}2S^2-o(S^2).
\]
In the secant argument, the decay exponent obtained for the lowered
threshold is proportional to this quantity divided by the squared
displacement $(x-1)^2S^2$, that is, to
\[
 f(x)=\frac{x^2-(1+\theta)}{(x-1)^2}.
\]
Since
$f'(x)=2(1+\theta-x)/(x-1)^3$, the function is maximized at $x=1+\theta$,
where it equals $(1+\theta)/\theta$.  This is the remote threshold
$S+T=(1+\theta)S$ used in Lemma~\ref{lem:asymptotic-two-threshold}, and
$(1+\theta)/\theta$ is the factor by which its decay exponent improves on
the exponent $a^2/2$ that follows from Borell's inequality alone.
\end{remark}

\section{A recursive cover step}
\label{sec:recursive-step}

Fix a root vertex $i$.  For any vertex $x$ with
$|\langle v_i,v_x\rangle|<1$, write
\begin{equation}
 h_x=\langle v_i,v_x\rangle,
 \qquad
 r_x=\sqrt{1-h_x^2},
 \qquad
 z_x=\frac{v_x-h_xv_i}{r_x}.
 \label{eq:endpoint-geometry}
\end{equation}
Here $h_x$ is the root correlation of $x$, and $z_x$ is the unit direction
of the component of $v_x$ orthogonal to the root.  A cover by the vectors
$z_x$ describes how the endpoints $x$ spread around the root.  All three
quantities depend only on the vertex $x$, not on the walk that reached it.
Consequently, a family of such vectors indexed by walks is the same
family as the one indexed by the distinct endpoints of those walks, and we
pass freely between the two descriptions; the packing conventions of
Section~\ref{sec:preliminaries} apply.

The following lemma contains all three neighborhood transitions used in
the proof.  It assumes a cover by the projected directions of the current
endpoints, a packing of next-edge directions attached to every current
endpoint, and a bound on the number of current endpoints.  It produces a
cover by the projected directions of the endpoints one edge farther away.
The lemma carries a fixed margin $\omega>0$, as explained in
Section~\ref{subsec:overview-round-or-continue}: the loss $\chi_\omega$ is
computed from $\theta+\omega$ and $\Sigma+\omega$ instead of $\theta$ and
$\Sigma$.  This makes every quantity that must be positive in the proof at
least $\omega$, and it includes the boundary cases $\theta=0$ and
$\Sigma=0$ without separate treatment.

\begin{lemma}[Recursive cover step]
\label{lem:recursive-cover-step}
Fix $c>0$, put $K_c=3(1+c)$, and fix a margin $\omega>0$.  Let
$S^\circ$ be a set of vertices such that
\[
 h_p=h+o(1)\qquad(p\in S^\circ),
 \qquad |h|<1,
 \qquad r=\sqrt{1-h^2}.
\]
Suppose that the vectors $\{z_p:p\in S^\circ\}$ form a
$(\Lambda t-o(t),\delta_0)$-cover with $\delta_0$ subexponential, and
that
\[
 \log|S^\circ|\leq\frac{\mu}{2}t^2+o(t^2).
\]
For each $p\in S^\circ$, let $\nu_p$ be a packing measure at threshold
$\sqrt3t$ for a family $\{v_{pm}:m\in J_p\}$ with $J_p\subseteq N_G(p)$
and
\[
 |J_p|\leq\exp\!\left(\frac{K_c}{2}t^2\right),
\]
and suppose that every $\nu_p$ has mass at least $\delta_1$, where
$\delta_1$ is subexponential.  Assume that, uniformly over all pairs
$p\in S^\circ$, $m\in J_p$,
\[
 \eta_{pm}:=\langle v_i,v_{pm}\rangle=\eta+o(1).
\]
Define
\begin{equation}
 \Sigma=\mu-\Lambda^2,
 \qquad
 \theta=c-(1+c)\frac{\eta^2}{r^2},
 \qquad
 \chi_\omega=\sqrt{\frac{(\theta+\omega)(\Sigma+\omega)}{1+c}},
 \label{eq:step-surplus-room-loss}
\end{equation}
and
\begin{equation}
 C=\frac r2+\tau\frac{h\eta}{r},
 \qquad
 h'=-\frac h2+\tau\eta.
 \label{eq:step-geometry}
\end{equation}
If $\Sigma\geq0$, $\theta\geq0$, $C\geq0$, $|h'|<1$, and
$\Lambda C-\tau\chi_\omega>0$, then the set
\[
 S'^\circ=\bigcup_{p\in S^\circ}J_p
\]
satisfies $h_m=h'+o(1)$ uniformly over $m\in S'^\circ$, and the vectors
$\{z_m:m\in S'^\circ\}$ form a
$(\Lambda'_\omega t-o(t),\ \delta_0/2)$-cover, where
\begin{equation}
 \Lambda'_\omega=
 \frac{\Lambda C+\frac32-\tau\chi_\omega}
      {\sqrt{1-(h')^2}}.
 \label{eq:step-recursion}
\end{equation}
The $o(\cdot)$ terms in the conclusion depend only on those in the
hypotheses, on the constants, and on $\omega$.
\end{lemma}

Here $\Sigma$ is the surplus of the outer cover, $\theta$ is the room left
for the projected inner vectors, and $\chi_\omega$ is the resulting loss;
see Section~\ref{subsec:overview-recursion}.  The numerator
in~\eqref{eq:step-recursion} has three contributions: $\Lambda C$ from the
outer cover, $3/2$ from the inner packing, and the loss $\tau\chi_\omega$
needed to make both crossings happen on the same walk.  The denominator
normalizes the new projected direction.

\begin{proof}
\textbf{Projecting the inner packings.}
Fix $p\in S^\circ$ and $m\in J_p$.  Since $v_p=h_pv_i+r_pz_p$ and
$\langle v_{pm},v_p\rangle=0$, we have
$\langle z_p,v_{pm}\rangle=-h_p\eta_{pm}/r_p$.  Hence the vector
\begin{equation}
 u_{pm}=v_{pm}-\eta_{pm}v_i+\frac{h_p\eta_{pm}}{r_p}z_p
 \label{eq:step-projection}
\end{equation}
is the projection of $v_{pm}$ orthogonal to
$\operatorname{span}\{v_i,z_p\}=\operatorname{span}\{v_i,v_p\}$, and a
direct expansion gives
\begin{equation}
 \|u_{pm}\|^2=1-\frac{\eta_{pm}^2}{r_p^2}.
 \label{eq:step-projected-norm}
\end{equation}
Put
\[
 q^2=1-\frac{\eta^2}{r^2}=\frac{1+\theta}{1+c}.
\]
The assumption $\theta\geq0$ gives $q>0$, and since $r_p=r+o(1)$ with
$r>0$, \eqref{eq:step-projected-norm} gives $\|u_{pm}\|=q+o(1)$ uniformly
over all pairs.

The projection removes two Gaussian coordinates, and we control them by
truncation.  Let $p_p=\nu_p(J_p)\geq\delta_1$ and put
\[
 \rho_p=\sqrt{2\log(16/p_p)}+1,
 \qquad
 d_{pm}=|\eta_{pm}|+\left|\frac{h_p\eta_{pm}}{r_p}\right|.
\]
By~\eqref{eq:step-projected-norm}, $|\eta_{pm}|\leq r_p$, and therefore
$d_{pm}\leq r_p+|h_p|\leq2$.  Define
\begin{equation}
 S_{p,t}=
 \min_{m\in J_p}
 \frac{\sqrt3t-\rho_pd_{pm}}{\|u_{pm}\|}.
 \label{eq:projected-threshold}
\end{equation}
The vectors $v_i$ and $z_p$ are orthogonal unit vectors, so
$\langle g,v_i\rangle$ and $\langle g,z_p\rangle$ are standard normal, and
by~\eqref{eq:gaussian-chernoff} the event
\[
 \mathcal T_p=\{|\langle g,v_i\rangle|\leq\rho_p,\
                 |\langle g,z_p\rangle|\leq\rho_p\}
\]
has probability at least $1-4\tail(\rho_p)\geq1-4e^{-\rho_p^2/2}
\geq1-p_p/4$.  On $\mathcal T_p$, a crossing
$\langle g,v_{pm}\rangle\geq\sqrt3t$ implies, by
\eqref{eq:step-projection},
\[
 \langle g,u_{pm}\rangle
 \geq\sqrt3t-\rho_pd_{pm}
 \geq\|u_{pm}\|S_{p,t}.
\]
The packing property~\eqref{eq:packing-measure} applied to $A=J_p$ gives
\[
 \Pr\!\left[
  \max_{m\in J_p}
  \left\langle g,\frac{u_{pm}}{\|u_{pm}\|}\right\rangle
  \geq S_{p,t}
 \right]
 \geq p_p-\Pr[\mathcal T_p^c]
 \geq\frac{3p_p}{4}
 \geq\frac{3\delta_1}{4}.
\]
Thus the normalized projected vectors form an
$(S_{p,t},3\delta_1/4)$-cover, and $3\delta_1/4$ is subexponential.
Since $\delta_1$ is subexponential, $\rho_p=o(t)$ uniformly in $p$, and
therefore
\begin{equation}
 S_{p,t}=\frac{\sqrt3}{q}t-o(t)
 \qquad\text{uniformly in }p.
 \label{eq:projected-calibration}
\end{equation}
In particular, $S_{p,t}/t$ eventually lies between two fixed positive
constants.

Define the room at finite $t$ by
\[
 \theta_t=
 \max_{p\in S^\circ}
 \left(\frac{2\log|J_p|}{S_{p,t}^2}-1\right)_+
 +\omega.
\]
Then $\theta_t\geq\omega$, and $\log|J_p|\leq(1+\theta_t)S_{p,t}^2/2$ for
every $p$.  The bound $|J_p|\leq\exp(K_ct^2/2)$ and
\eqref{eq:projected-calibration} give
\[
 \theta_t
 \leq\left(\frac{K_ct^2}{S_{p,t}^2}-1\right)_++\omega
 \leq\frac{K_cq^2}{3}-1+\omega+o(1)
 =\theta+\omega+o(1).
\]

\medskip
\noindent\textbf{Composing the outer and inner covers.}
Let $s_t=\Lambda t-o(t)$ be the actual outer threshold and define the
surplus at finite $t$ by
\[
 b_t=\frac2{t^2}\Bigl(\log\bigl(|S^\circ|\,\tail(s_t)\bigr)\Bigr)_+,
 \qquad
 \Sigma_t=b_t+\omega.
\]
The cardinality bound on $S^\circ$ and~\eqref{eq:tail-log-uniform} give
$\log(|S^\circ|\tail(s_t))\leq(\mu-\Lambda^2)t^2/2+o(t^2)$, and hence
$b_t\leq\Sigma+o(1)$; also $\Sigma_t\geq\omega$.  Put
\[
 a_t=\sqrt{\frac{\theta_t\Sigma_t}{1+\theta_t}}\,t
 \leq\sqrt{\Sigma+\omega+1}\;t
\]
for all large $t$, and define the inner-failure events
\[
 \mathcal F_p=
 \left\{
  \max_{m\in J_p}
  \left\langle g,\frac{u_{pm}}{\|u_{pm}\|}\right\rangle
  \leq S_{p,t}-a_t
 \right\},
 \qquad p\in S^\circ.
\]
Apply Lemma~\ref{lem:asymptotic-two-threshold} simultaneously to all the
normalized projected families, with $\theta_0=\omega$, with
$M=\theta+\omega+1$, with $A=\sqrt{\Sigma+\omega+1}$, with the common
value $\theta_t$, with $R_t=0$, and with $L_t=\log(4/(3\delta_1))$.  All
hypotheses hold uniformly by the preceding paragraph.  The conclusion, with
$a=a_t$, gives, uniformly in $p$,
\[
 \log\Pr(\mathcal F_p)
 \leq-\frac{1+\theta_t}{2\theta_t}\cdot\frac{\theta_t\Sigma_t}{1+\theta_t}t^2
     +o(t^2)
 =-\frac{b_t}{2}t^2-\frac\omega2t^2+o(t^2)
 \leq-\frac{b_t}{2}t^2-\frac\omega4t^2
\]
for all large $t$.

Let $\mathcal A_p=\{\langle g,-z_p\rangle\geq s_t\}$, so that
$\Pr(\mathcal A_p)=\tail(s_t)$.  The event $\mathcal F_p$ depends only on
the projection of $g$ onto $\operatorname{span}\{v_i,z_p\}^\perp$ and is
therefore independent of $\mathcal A_p$.  By the definition of $b_t$,
\[
 \Pr\!\left[\bigcup_{p\in S^\circ}(\mathcal A_p\cap\mathcal F_p)\right]
 \leq\sum_{p\in S^\circ}\Pr(\mathcal A_p)\Pr(\mathcal F_p)
 \leq|S^\circ|\,\tail(s_t)\,e^{-b_tt^2/2-\omega t^2/4}
 \leq e^{-\omega t^2/4}
 \leq\frac{\delta_0}{2}
\]
for all large $t$, because $\delta_0$ is subexponential.  On the other
hand, the vectors $-z_p$ have the same covering probability as the vectors
$z_p$, so $\Pr[\bigcup_p\mathcal A_p]\geq\delta_0$.  Hence, with
probability at least $\delta_0/2$, some $\mathcal A_p$ occurs while
$\mathcal F_p$ does not.

We convert the inner success back to the unnormalized vectors.  If
$\mathcal F_p$ does not occur, some $m\in J_p$ satisfies
$\langle g,u_{pm}\rangle\geq\|u_{pm}\|(S_{p,t}-a_t)$.  By
\eqref{eq:projected-calibration}, $\|u_{pm}\|S_{p,t}=\sqrt3t-o(t)$.  For
the loss, note that $x\mapsto x/(1+x)$ is increasing, so
$\theta_t\leq\theta+\omega+o(1)$ gives
\[
 \|u_{pm}\|^2\frac{\theta_t}{1+\theta_t}
 \leq\frac{1+\theta}{1+c}\cdot
     \frac{\theta+\omega+o(1)}{1+\theta+\omega+o(1)}
 \leq\frac{\theta+\omega+o(1)}{1+c},
\]
and $\Sigma_t\leq\Sigma+\omega+o(1)$.  Hence
$\|u_{pm}\|a_t\leq\chi_\omega t+o(t)$, and
\begin{equation}
 \langle g,u_{pm}\rangle\geq(\sqrt3-\chi_\omega)t-o(t).
 \label{eq:inner-success}
\end{equation}

\medskip
\noindent\textbf{Passing to the endpoints.}
Let $m\in J_p$.  The edge identity~\eqref{eq:edge-identity},
$v_m=-v_p/2+\tau v_{pm}$, and~\eqref{eq:step-projection} give
\begin{align}
 h_m=\langle v_i,v_m\rangle
 &= -\frac{h_p}{2}+\tau\eta_{pm}=h'+o(1),
 \label{eq:step-endpoint-correlation}\\
 v_m-h_mv_i
 &= -\left(\frac{r_p}{2}
       +\tau\frac{h_p\eta_{pm}}{r_p}\right)z_p+\tau u_{pm}
  = -C_pz_p+\tau u_{pm},
 \qquad C_p=C+o(1).
 \label{eq:step-endpoint-vector}
\end{align}
The first line holds uniformly over $m\in S'^\circ$, which is the first
assertion.  Now suppose that $\mathcal A_p$ occurs and $\mathcal F_p$ does
not, and let $m$ be as in~\eqref{eq:inner-success}.  Then
\[
 \langle g,v_m-h_mv_i\rangle
 =C_p\langle g,-z_p\rangle+\tau\langle g,u_{pm}\rangle
 \geq C_ps_t+\tau(\sqrt3-\chi_\omega)t-o(t)
 \geq\left(\Lambda C+\frac32-\tau\chi_\omega\right)t-o(t).
\]
In the last step we used $C_ps_t\geq\Lambda Ct-o(t)$, which holds because
$C\geq0$ and $C_p=C+o(1)$; here $3/2=\tau\sqrt3$.  Since $|h'|<1$ and
$h_m=h'+o(1)$, the norm $r_m=\sqrt{1-h_m^2}$ equals
$\sqrt{1-(h')^2}+o(1)$ and is bounded away from zero, so dividing by
$r_m$ gives $\langle g,z_m\rangle\geq\Lambda'_\omega t-o(t)$.  Thus, with
probability at least $\delta_0/2$, some $m\in S'^\circ$ satisfies this
inequality, which is the claimed cover.  All choices were made before $g$
was sampled, so $S'^\circ$ is deterministic.
\end{proof}

The hypothesis $\theta\geq0$ says that the next-edge directions are not
too correlated with the root: $|\eta|\leq\rho_cr$ with
$\rho_c=\sqrt{c/(1+c)}$.  At the last transition, some attached directions
may violate this bound.  The next lemma shows that they carry very little
packing mass and can be discarded.

\begin{lemma}[Discarding overcorrelated directions]
\label{lem:discard-overcorrelated}
Fix $c>0$, put $K_c=3(1+c)$ and $\rho_c=\sqrt{c/(1+c)}$, and fix
$r_{\min}>0$.  Let $p$ be a vertex with $r_p\geq r_{\min}$, and let
$\nu_p$ be a packing measure at threshold $\sqrt3t$ for a family
$\{v_{pm}:m\in J_p\}$ with $J_p\subseteq N_G(p)$,
$|J_p|\leq\exp(K_ct^2/2)$, and mass $\delta$, where $\delta=\delta_t$ is
subexponential.  There is a sequence $\epsilon_t\to0$, depending only on
$c$, $r_{\min}$, and $\delta_t$, such that for all large $t$
\[
 \nu_p\bigl\{m\in J_p:\ |\eta_{pm}|\geq\rho_cr_p+\epsilon_t\bigr\}
 \leq\frac\delta4.
\]
\end{lemma}

\begin{proof}
Choose a constant $\sigma_0$ with $0<\sigma_0<\tau\rho_c(1+c)$ and put
\[
 F(\epsilon)
 =\frac{(\sqrt3-2\sigma_0\epsilon)^2}
       {1-(\rho_c+\epsilon)^2}-K_c .
\]
Since $1-\rho_c^2=1/(1+c)$, we have $F(0)=3(1+c)-K_c=0$ and
\[
 F'(0)=6\rho_c(1+c)^2-4\sqrt3\sigma_0(1+c)>0 .
\]
Hence there are constants $k_0>0$ and $\epsilon_0\in(0,1-\rho_c)$ with
$F(\epsilon)\geq k_0\epsilon$ for $0\leq\epsilon\leq\epsilon_0$.  Put
$L=\log(32/\delta)$ and
\[
 \epsilon_t=t^{-1/2}
 +\left(\frac2{\sigma_0}+\frac2{k_0}\right)\frac{\sqrt L}{t}.
\]
Since $L=o(t^2)$, we have $\epsilon_t\to0$, so $\epsilon_t\leq\epsilon_0$
for all large $t$.

Let $\mathcal B_p=\{m\in J_p:|\eta_{pm}|\geq\rho_cr_p+\epsilon_t\}$.  For
$m\in\mathcal B_p$, \eqref{eq:step-projected-norm} and $r_p\leq1$ give
\[
 \|u_{pm}\|^2=1-\frac{\eta_{pm}^2}{r_p^2}
 \leq1-(\rho_c+\epsilon_t)^2 .
\]
Outside the event
$\{|\langle g,v_i\rangle|>\sigma_0\epsilon_tt\}
 \cup\{|\langle g,z_p\rangle|>\sigma_0\epsilon_tt\}$, whose probability
is at most $4\tail(\sigma_0\epsilon_tt)$, a crossing
$\langle g,v_{pm}\rangle\geq\sqrt3t$ with $m\in\mathcal B_p$ forces, by
\eqref{eq:step-projection} and $d_{pm}\leq2$,
\[
 \left\langle g,\frac{u_{pm}}{\|u_{pm}\|}\right\rangle
 \geq\frac{(\sqrt3-2\sigma_0\epsilon_t)t}
          {\sqrt{1-(\rho_c+\epsilon_t)^2}}
 =:s'_t,
 \qquad
 (s'_t)^2=(K_c+F(\epsilon_t))t^2\geq(K_c+k_0\epsilon_t)t^2 .
\]
The packing property, the union bound over $\mathcal B_p$, and
\eqref{eq:gaussian-chernoff} give
\[
 \nu_p(\mathcal B_p)
 \leq4\tail(\sigma_0\epsilon_tt)+|J_p|\tail(s'_t)
 \leq4\exp\!\left(-\frac{\sigma_0^2\epsilon_t^2t^2}{2}\right)
  +\exp\!\left(\frac{K_c}{2}t^2-\frac{K_c+k_0\epsilon_t}{2}t^2\right).
\]
By the choice of $\epsilon_t$, $\sigma_0^2\epsilon_t^2t^2/2\geq2L$, so the
first term is at most $4e^{-2L}\leq\delta/8$.  For the second term,
$k_0\epsilon_tt^2/2\geq t\sqrt L\geq L$ once $L\leq t^2$, which holds
for all large $t$; hence the second term is at most $e^{-L}\leq\delta/8$.
\end{proof}

At the end of each level, we compare the number of candidate endpoints
with a cutoff padded by the margin.  A set above the cutoff is rounded,
while a set below it has the size bound needed for the next cover step.

\begin{lemma}[Cardinality split]
\label{lem:cardinality-split}
Fix $c>0$, put $K_c=3(1+c)$, fix a margin $\omega>0$, and for
$\mu\geq0$ define the cutoff
\begin{equation}
 N_t(\mu)=
 \left\lceil
   \exp\!\left(\frac{\mu+\omega}{2}t^2\right)
 \right\rceil .
 \label{eq:cutoff}
\end{equation}
Let $S^\circ\subseteq\widehat S\subseteq V$, and suppose that a vector
$\kappa_t$-coloring of $G[\widehat S]$ can be constructed in polynomial
time, where $2\leq\kappa_t\leq K$ for a fixed $K$ and
$1-2/\kappa_t\leq\pi_t+o(1)$ for a bounded sequence $\pi_t\geq0$.  Suppose
also that
\[
 \Delta(G)\leq D,
 \qquad
 \log D=\Theta(t^2),
 \qquad
 \log D\leq\frac{K_c}{2}t^2 .
\]
Then exactly one of the following holds.
\begin{enumerate}
\item $|\widehat S|\geq N_t(\mu)$, and KMS rounding on $G[\widehat S]$
 finds, with high probability, an independent set of size at least
 $D^{(\mu+\omega)/K_c-\pi_t-o(1)}\geq D^{\mu/K_c-\pi_t-o(1)}$.
\item $|\widehat S|<N_t(\mu)$, and hence
 \[
  \log|S^\circ|\leq\log|\widehat S|<\frac{\mu+\omega}{2}t^2 .
 \]
 In particular, $S^\circ$ satisfies the cardinality hypothesis of
 Lemma~\ref{lem:recursive-cover-step} with coefficient $\mu+\omega$.
\end{enumerate}
\end{lemma}

\begin{proof}
The two cases are exhaustive because $|\widehat S|$ and $N_t(\mu)$ are
integers.  In the first case,
$|\widehat S|\geq\exp((\mu+\omega)t^2/2)\geq D^{(\mu+\omega)/K_c}$.
Theorem~\ref{thm:kms-extraction} applied to $G[\widehat S]$, whose maximum
degree is at most $D$, gives an independent set of size
$\Omega(|\widehat S|D^{-(1-2/\kappa_t)}/\sqrt{\log(2D)})$, and
$\sqrt{\log(2D)}=D^{o(1)}$ because $\log D\to\infty$.  In the second case,
the definition of the ceiling gives
$|\widehat S|<\exp((\mu+\omega)t^2/2)$.
\end{proof}

\section{Structural input from a failed rounding}
\label{sec:structural-input}

This section proves the proposition that starts the recursion.  It says
that a failed rounding yields a root whose neighbor directions form a
packing, together with packings of next-edge directions attached to the
retained walks of length one, two, and three, with the correlations along
these walks confined to prescribed ranges.  The proof adapts the method of
\citet*{BHL26} to a graph of maximum degree at most $D$.  It quotes two
packing theorems from that paper and supplies all remaining arguments.

\begin{definition}[Inefficiency and spread]
An $(s,\delta)$-packing $(X,\nu)$ is \emph{at most $a$-inefficient} if
\[
 |X|\leq\tail(s)^{-(1+a)}.
\]
For a set $V_0$ of unit vectors, it is $(\lambda,p,V_0)$-\emph{spread} if,
for every $v\in V_0$,
\[
 \nu\{x\in X:\langle v,x\rangle\geq\lambda\}\leq p.
\]
\end{definition}

Thus an efficient packing has few vectors relative to its threshold, and a
spread packing puts little mass near any direction in $V_0$.

\begin{theorem}[Packing-spread theorem, Theorem~4.11 of \citet*{BHL26}]
\label{thm:bhl-packing-spread}
Fix $a>0$.  Let $(X,\nu)$ be an at most $a$-inefficient
$(s,\delta)$-packing, where $s\to\infty$.  There is a restriction
$X'\subseteq X$ such that
\[
 \nu(X')\geq\nu(X)-\frac1{\log s},
\]
and $(X',\nu|_{X'})$ is $(\lambda_s,p_s,X')$-spread, where
\[
 \lambda_s=\frac{a}{1+a}+O\!\left(\frac1{\log s}\right),
 \qquad
 p_s=\exp[-\Omega(\log^2s)].
\]
The constants in the asymptotic terms depend only on $a$.
\end{theorem}

The displayed mass conclusion follows from the published statement by
taking its packing parameter to be the actual total mass $\nu(X)$, which is
legitimate because the first condition in~\eqref{eq:packing-measure},
applied to all of $X$, makes $(X,\nu)$ an $(s,\nu(X))$-packing.

\begin{theorem}[Cross-packing theorem, Corollary~5.4 of \citet*{BHL26}]
\label{thm:bhl-cross-packing}
Fix $a>0$.  For $r=1,2$, let $(X_r,\nu_r)$ be an at most
$a$-inefficient $(s,\delta_r)$-packing, where $s\to\infty$.  There is a
restriction $X'_1\subseteq X_1$ with
\[
 \nu_1(X'_1)\geq\delta_1-O(1/\log s)
\]
such that $(X_2,\nu_2)$ is
$(\lambda_s,p_s,X'_1\cup(-X'_1))$-spread, where $\lambda_s$ and $p_s$ are
as in Theorem~\ref{thm:bhl-packing-spread}.  Again the constants depend
only on $a$.
\end{theorem}

\begin{proposition}[Structural input from a failed rounding]
\label{prop:structural-input}
Fix $0<\widehat c<c$ and put $r_c=c/(1+c)$.  Let $G$ be an $n$-vertex
$3$-colorable graph equipped with a solution to the three-round relaxation
of Section~\ref{subsec:overview-sdp}.  Suppose $D\to\infty$ and
$\Delta(G)\leq D$, and choose $t$ by
\[
 \tail(t)=D^{-1/[3(1+\widehat c)]}.
\]
If the modified KMS rounding fails, then there exist a root $i$, a set
$J\subseteq N(i)$, sets $J_j\subseteq N(j)$ for $j\in J$, sets
$L_k\subseteq N(k)$ for $k\in W:=\bigcup_{j\in J}J_j$, and sets
$M_\ell\subseteq N(\ell)$ for $\ell\in\bigcup_{k\in W}L_k$, with the
following properties.  Every family below carries a subexponential packing
at threshold $\sqrt3t$; in fact all masses are at least
$(\log\log t)^{-3}$ for all large $t$.
\begin{enumerate}
\item[(a)] The directions $\{v_{ij}:j\in J\}$ form such a packing.
\item[(b)] For every $j\in J$, the directions $\{v_{jk}:k\in J_j\}$ form
 such a packing, and every $k\in J_j$ satisfies
 \[
  -o(1)\leq\langle v_{ji},v_{jk}\rangle\leq r_c+o(1).
 \]
\item[(c)] For every $k\in W$, the directions $\{v_{k\ell}:\ell\in L_k\}$
 form such a packing, and every $\ell\in L_k$ satisfies
 \[
  -\frac{3\sqrt3}{4}r_c-o(1)
  \leq\langle v_i,v_{k\ell}\rangle
  \leq\frac{\sqrt3}{2}r_c+o(1).
 \]
\item[(d)] For every $\ell\in\bigcup_{k\in W}L_k$, the directions
 $\{v_{\ell m}:m\in M_\ell\}$ form such a packing.
\end{enumerate}
Every family has at most $D$ members, and all $o(1)$ terms are uniform
over the retained families.
\end{proposition}

The recursion of Section~\ref{sec:three-levels} uses only this output.
Part~(a) is the level-one cover, parts~(b) and~(c) supply the next-edge
packings for the transitions to levels two and three, and part~(d)
supplies them for the transition to level four.  The correlation range in
(b) determines the level-two root correlation through
$\langle v_i,v_k\rangle=1/4+3\langle v_{ji},v_{jk}\rangle/4$, and the
range in~(c) is the range of the parameter $\gamma$ at level three.

\begin{proof}
Put $s=\sqrt3t$, the common packing threshold, and
$L_s=\log(1/\tail(s))$.  By~\eqref{eq:tail-log-uniform},
$\log D=3(1+\widehat c)(t^2/2+O(\log t))$ and $L_s=3t^2/2+O(\log t)$, so
$\log D=(1+\widehat c+o(1))L_s$.  Because $c-\widehat c>0$ is fixed, for
all sufficiently large $n$ we have
\begin{equation}
 2D\leq\tail(s)^{-(1+c)}.
 \label{eq:degree-cap-inefficiency}
\end{equation}
Every family of neighbor directions occurring below has at most $D$
members, so every such family, and every sign-symmetrized family
$X\cup(-X)$, is at most $c$-inefficient at threshold $s$.  We apply the
two packing theorems with $a=c$ throughout.  Their asymptotic constants are
then uniform, and we fix constants $C_\lambda,C_\mu,c_q>0$ such that every
application below has correlation cutoff at most
\[
 \lambda_t=r_c+\frac{C_\lambda}{\log s},
\]
exceptional mass at most $p_t=\exp[-c_q\log^2s]$, and cross-packing
restriction loss at most $C_\mu/\log s$.  Increasing a cutoff or an
exceptional-mass bound preserves a spread conclusion, so these common values
may be fixed in advance.

\medskip
\noindent\textbf{Packing measures from the matching.}
For a standard Gaussian vector $g$, let $M(g)$ be the maximal matching
chosen by the fixed tie-breaking rule in the graph induced by the selected
vertices $\{x:\langle g,v_x\rangle\geq t\}$.  For every edge $\{x,y\}$ and
every vertex $x$, set
\begin{equation}
 w(\{x,y\})=
 \frac{\Pr[\{x,y\}\in M(g)]}{\tail(t)},
 \qquad
 \mu_x(A)=\sum_{y\in A}w(\{x,y\})
 \quad(A\subseteq N(x)).
 \label{eq:matching-weight}
\end{equation}
The weights are symmetric.  For fixed $x$, the events $\{x,y\}\in M(g)$,
$y\in N(x)$, are disjoint, and each of them requires that both $x$ and $y$
are selected, which by~\eqref{eq:edge-identity} implies
\[
 \langle g,v_{xy}\rangle
 =\frac2{\sqrt3}\left(
   \langle g,v_y\rangle+\frac12\langle g,v_x\rangle\right)
 \geq s.
\]
The family $\{v_{xy}:y\in A\}$ lies in $v_x^\perp$, so the event that one
of these directions crosses $s$ is independent of the event
$\langle g,v_x\rangle\geq t$.  Hence
\[
 \mu_x(A)
 \leq\frac{\Pr[\langle g,v_x\rangle\geq t\text{ and }
             \exists y\in A:\langle g,v_{xy}\rangle\geq s]}{\tail(t)}
 =\Pr[\exists y\in A:\langle g,v_{xy}\rangle\geq s].
\]
Thus $\mu_x$ is a packing measure at threshold $s$ for the neighbor
directions of $x$, and its total mass $\mu_x(N(x))$ is the conditional
probability that $x$ is matched, and hence deleted, given that it is
selected.  Since the rounding fails, this conditional probability is at
least $1/2$ for at least $n/2$ vertices, and therefore
\begin{equation}
 \sum_{e\in E}w(e)
 =\frac12\sum_{x\in V}\mu_x(N(x))\geq\frac n8.
 \label{eq:initial-weight}
\end{equation}

\medskip
\noindent\textbf{Pruning to a core.}
We use the following standard observation.  If a weighted graph on $N$
vertices has total edge weight at least $rN$, then repeatedly deleting a
vertex of incident weight less than $r$ cannot delete the whole graph,
because the total deleted weight would be less than $rN$.  The remaining
nonempty subgraph has incident weight at least $r$ at every vertex.

Apply this to~\eqref{eq:initial-weight} with $r=1/16$, and call the
resulting graph $H_0=(V_0,E_0)$.  Every vertex of $H_0$ has incident
packing mass at least $1/16$, and if $n-|V_0|$ vertices were deleted, the
weight left in $H_0$ is at least
\[
 \frac n8-\frac{n-|V_0|}{16}=\frac{n+|V_0|}{16}\geq\frac{|V_0|}{8}.
\]
For every $x\in V_0$, apply Theorem~\ref{thm:bhl-packing-spread} to the
packing $\mu_x$ on the neighbor directions of $x$ in $H_0$, and let $X_x$
be the restricted family.  Keep an edge $\{x,y\}$ of $H_0$ only if
$v_{xy}\in X_x$ and $v_{yx}\in X_y$.  The total deleted weight is at most
$|V_0|/\log s$, so the remaining graph $H_1$ has total weight at least
$|V_0|(1/8-o(1))$.

Put $L=\log\log t$.  Starting from $H_1$, perform the following pruning
until no rule applies; here $N(j)$ denotes the neighborhood of $j$ in the
current graph, so the mass in the rule is recomputed after every deletion.
If an oriented edge $(i,j)$ satisfies
\[
 \mu_j\{k\in N(j):\langle v_{ji},v_{jk}\rangle\geq-1/L\}<L^{-2},
\]
delete every edge $\{j,k\}$ with $k$ in this set.  Fix a vertex $j$, and
list as $i_1,\ldots,i_q$ the neighbors that trigger this rule at $j$, in
chronological order.  For $a<b$, the edge $\{j,i_b\}$ survived the rule
triggered by $i_a$, so $\langle v_{ji_a},v_{ji_b}\rangle<-1/L$.  Hence
\[
 0\leq\left\|\sum_{a=1}^qv_{ji_a}\right\|^2
 <q-\frac{q(q-1)}L,
\]
so $q<L+1$.  Each trigger deletes incident weight less than $L^{-2}$, so
this stage deletes total weight $O(|V_0|/L)$.  Denote the final graph by
$G'$.  For all large $t$, it has total weight at least $|V_0|/9$ and is
therefore nonempty.

At termination, every oriented edge $(i,j)$ of $G'$ carries the packing
\begin{equation}
 P_{ij}=\bigl(\{v_{jk}:k\in I_{ij}\},\ \mu_j|_{I_{ij}}\bigr),
 \qquad
 I_{ij}=\{k\in N_{G'}(j):
  -1/L\leq\langle v_{ji},v_{jk}\rangle\leq\lambda_t\}.
 \label{eq:oriented-packing}
\end{equation}
Indeed, since no pruning rule applies, the mass above the lower cutoff
$-1/L$ is at least $L^{-2}$, while the spread property of $X_j$ puts mass
at most $p_t$ above $\lambda_t$.  Thus $P_{ij}$ has mass at least
$L^{-2}-p_t\geq1/(2L^2)$.  We write $v_{jk}\in P_{ij}$ for $k\in I_{ij}$,
and a restriction of $P_{ij}$ means the corresponding restriction of both
its indexed family and its measure.  Note that the condition defining
$I_{ij}$ is symmetric in $i$ and $k$.

\medskip
\noindent\textbf{Making successive packings compatible.}
For every oriented edge $(k,j)$ of $G'$, apply
Theorem~\ref{thm:bhl-cross-packing} with first packing $P_{kj}$ and second
packing $P_{jk}$; thus the theorem is invoked in both orientations of each
edge, and~\eqref{eq:degree-cap-inefficiency} verifies its cardinality
hypothesis.  We obtain a restriction $P'_{kj}\subseteq P_{kj}$ of mass at
least
\[
 \delta_t:=\frac1{3L^2},
\]
for all large $t$, because $1/(2L^2)-C_\mu/\log s\geq1/(3L^2)$, such that
$P_{jk}$ is $(\lambda_t,p_t,P'_{kj}\cup(-P'_{kj}))$-spread.  Say that $k$
is \emph{good for} $(i,j)$ when $v_{ji}\in P'_{kj}$.  For every good triple
define
\[
 U_{ijk}=\{\ell:v_{k\ell}\in P_{jk},\
                  |\langle v_{ji},v_{k\ell}\rangle|\leq\lambda_t\}.
\]
The two signed spread bounds discard mass at most $2p_t$, so $U_{ijk}$ has
packing mass at least $1/(2L^2)-2p_t\geq\delta_t$.  Moreover, whenever $k$
is good for $(i,j)$ and $\ell\in U_{ijk}$,
\begin{align}
 -1/L&\leq\langle v_{ji},v_{jk}\rangle\leq\lambda_t,
 \label{eq:good-first-correlation}\\
 -1/L&\leq\langle v_{kj},v_{k\ell}\rangle\leq\lambda_t,
 \qquad
 |\langle v_{ji},v_{k\ell}\rangle|\leq\lambda_t.
 \label{eq:good-cross-correlations}
\end{align}
The first line holds because $v_{ji}\in P_{kj}$ means $i\in I_{kj}$, which
is the same condition as $k\in I_{ij}$.

\medskip
\noindent\textbf{Choosing the root.}
For an oriented edge $(i,j)$ of $G'$ put
\[
 m_{ij}=\mu_j\{k\in N_{G'}(j):k\text{ is good for }(i,j)\}.
\]
By the symmetry of the weights,
\begin{align*}
 \sum_{i\in V_0}\sum_{j\in N_{G'}(i)}\mu_i(j)\,m_{ij}
 &=\sum_{j\in V_0}\sum_{k\in N_{G'}(j)}\mu_j(k)\,
   \mu_j\{i\in N_{G'}(j):k\text{ is good for }(i,j)\}\\
 &=\sum_{j\in V_0}\sum_{k\in N_{G'}(j)}\mu_j(k)\,\nu(P'_{kj})
 \geq\delta_t\sum_{j\in V_0}\mu_j(N_{G'}(j))
 \geq\frac{2\delta_t}{9}|V_0|,
\end{align*}
where $\nu(P'_{kj})$ denotes the mass of $P'_{kj}$ and the last step uses
the total weight of $G'$.  Hence some vertex $i$ satisfies
$\sum_{j}\mu_i(j)m_{ij}\geq2\delta_t/9$.  Fix this root and let
\[
 J=\{j\in N_{G'}(i):m_{ij}\geq\delta_t/10\},
 \qquad
 J_j=\{k\in N_{G'}(j):k\text{ is good for }(i,j)\}
 \quad(j\in J).
\]
Every $m_{ij}$ is at most one and $\mu_i$ has total mass at most one, so
\[
 \mu_i(J)
 \geq\frac{2\delta_t}{9}-\frac{\delta_t}{10}
 =\frac{11\delta_t}{90}.
\]
This is part~(a), and $\mu_j(J_j)=m_{ij}\geq\delta_t/10$ together with
\eqref{eq:good-first-correlation} is part~(b).

\medskip
\noindent\textbf{The attached packings.}
For each distinct $k\in W=\bigcup_{j\in J}J_j$, choose one $j=j(k)\in J$
with $k\in J_j$, and set $L_k=U_{ij(k)k}$.  Its packing mass is at least
$\delta_t$.  Along the walk $i\to j\to k$, the edge
identity~\eqref{eq:edge-identity} applied twice gives
\[
 v_i=\frac14v_k-\frac{\sqrt3}{4}v_{kj}+\frac{\sqrt3}{2}v_{ji}.
\]
Since $v_{k\ell}\perp v_k$, equations~\eqref{eq:good-cross-correlations}
give
\[
 \langle v_i,v_{k\ell}\rangle
 =-\frac{\sqrt3}{4}\langle v_{kj},v_{k\ell}\rangle
  +\frac{\sqrt3}{2}\langle v_{ji},v_{k\ell}\rangle
 \in\left[-\frac{3\sqrt3}{4}\lambda_t,\
          \frac{\sqrt3}{2}\lambda_t+\frac{\sqrt3}{4L}\right],
\]
which is the range in part~(c).  Finally, every $\ell\in L_k$ satisfies
$v_{k\ell}\in P_{jk}$, so $\{k,\ell\}$ is an edge of $G'$.  For each
distinct $\ell\in\bigcup_kL_k$ choose one such $k=k(\ell)$ and let
$M_\ell=I_{k(\ell)\ell}$, so that $\{v_{\ell m}:m\in M_\ell\}$ carries the
packing $P_{k(\ell)\ell}$ of~\eqref{eq:oriented-packing}, of mass at least
$1/(2L^2)$.  This is part~(d).

All displayed masses are at least $\delta_t/10=1/(30L^2)\geq L^{-3}$ for
all large $t$, every family has at most $D$ members, and all error terms
are $O(1/L)+O(1/\log s)=o(1)$, uniformly, because $c$ is fixed.
\end{proof}

\section{The three levels}
\label{sec:three-levels}

This section runs the recursion of Lemma~\ref{lem:recursive-cover-step}
through its three transitions and proves the bounded-degree guarantee,
Proposition~\ref{prop:failed-rounding-progress}.  We first fix the
constants and write out the recursion formulas at each level, then state
the numerical inequalities that the argument needs, and finally give the
proof.

\subsection{Constants and the recursion at each level}
\label{subsec:constants}

The constants are
\begin{equation}
\begin{aligned}
 c&=\frac{207}{1000},
 &\widehat c&=\frac{2067}{10000},
 &K_c&=3(1+c),\\
 r_c&=\frac{c}{1+c},
 &\rho_c&=\sqrt{\frac{c}{1+c}},
 &\tau&=\frac{\sqrt3}2,\\
 \mu_2^*&=\frac{2291}{2000}
 \left(\frac{2009199}{10^6}\right)^2,
 &A_0&=\frac{3191}{2500}.
\end{aligned}
\label{eq:analytic-parameters}
\end{equation}
For later reference put
\[
 R(x)=\frac{4+3x}{3+3x},
 \qquad
 \Xi(x)=\frac{5+3x}{3+3x}.
\]
The constant $\widehat c<c$ determines the Gaussian threshold and the
final exponent, while the recursion is run with $c$.  The gap between them
is what makes the degree bound $D$ strictly smaller than
$\exp(K_ct^2/2)$, as required by Lemma~\ref{lem:recursive-cover-step}.

\paragraph{The one-step map.}
For a margin $\omega\geq0$, Lemma~\ref{lem:recursive-cover-step} takes the
data $(\Lambda,\mu,h,\eta)$ of the current level and returns the data
$(\Lambda'_\omega,h')$ of the next level through the formulas
\eqref{eq:step-surplus-room-loss}--\eqref{eq:step-recursion}.  We now
write these formulas out at each level.  Every quantity below with a
subscript $\omega$ is defined by exactly the formulas of the lemma; the
quantities without a subscript are their values at $\omega=0$, which are
the ones the certificate in Appendix~\ref{sec:numerics} evaluates.

\paragraph{Level one.}
The root is $i$, and the level-one endpoints are its neighbors $j$.  By
\eqref{eq:edge-identity},
\[
 h_1=-\frac12,
 \qquad
 r_1=\tau,
 \qquad
 z_j=v_{ij},
 \qquad
 \Lambda_1=\sqrt3,
 \qquad
 \mu_1=K_c,
 \qquad
 \Sigma_1=\mu_1-\Lambda_1^2=3c.
\]
The threshold $\Lambda_1t=\sqrt3t$ is the packing threshold of
Proposition~\ref{prop:structural-input}(a), and $\mu_1=K_c$ comes from the
degree bound $|J|\leq D\leq\exp(K_ct^2/2)$.

\paragraph{Transition to level two.}
For a walk $i\to j\to k$, put $\alpha=\langle v_{ji},v_{jk}\rangle$.  Since
$v_i=-v_j/2+\tau v_{ji}$ and $v_{jk}\perp v_j$, the next-edge correlation
is $\eta_2=\langle v_i,v_{jk}\rangle=\tau\alpha$, and taking the inner
product of $v_i$ and $v_k=-v_j/2+\tau v_{jk}$ gives the exact identity
\begin{equation}
 \beta:=\langle v_i,v_k\rangle=\frac14+\frac34\alpha,
 \qquad\text{equivalently}\qquad
 \alpha=\frac{4\beta-1}{3}.
 \label{eq:alpha-beta}
\end{equation}
In particular, $\alpha$ depends only on the endpoint $k$.  By
Proposition~\ref{prop:structural-input}(b), the limiting value of $\alpha$
lies in $[0,r_c]$, so $\beta$ lies in
\[
 B_c=\left[\frac14,\ \frac14+\frac{3r_c}{4}\right].
\]
The formulas of the lemma become
\begin{equation}
\begin{aligned}
 \theta_2(\alpha)&=c-(1+c)\alpha^2,
 &C_2&=\frac{\tau(1-\alpha)}{2},
 &h_2&=\beta,\\
 \chi_{2,\omega}(\alpha)&=\sqrt{\frac{(\theta_2(\alpha)+\omega)(3c+\omega)}{1+c}},
 &\Lambda_{2,\omega}(\beta)
 &=\frac{\frac34(1-\alpha)+\frac32-\tau\chi_{2,\omega}(\alpha)}
        {\sqrt{1-\beta^2}}.
\end{aligned}
\label{eq:level-two-threshold}
\end{equation}
At $\omega=0$, writing $\ell_2(\alpha)=\sqrt{c\,\theta_2(\alpha)/(1+c)}$,
\[
 \Lambda_2(\beta)
 =\frac{9-3\alpha-6\ell_2(\alpha)}{4\sqrt{1-\beta^2}}.
\]

\paragraph{Level two.}
The level-two endpoints have root correlation $h_2=\beta+o(1)$ and cover
threshold $\Lambda_{2,\omega}(\beta)t-o(t)$.  If their candidate set is
below its cutoff, Lemma~\ref{lem:cardinality-split} gives the cardinality
coefficient $\mu_2=\mu_2^*+\omega$, and the surplus is
\begin{equation}
 \Sigma_{2,\omega}(\beta)=\mu_2^*+\omega-\Lambda_{2,\omega}(\beta)^2 .
 \label{eq:sigma-two}
\end{equation}

\paragraph{Transition to level three.}
For a level-two endpoint $k$ and a next edge $\{k,\ell\}$, put
$\gamma=\langle v_i,v_{k\ell}\rangle$.  By
Proposition~\ref{prop:structural-input}(c), the limiting value of $\gamma$
lies in
\[
 \Gamma_c=\left[-\frac{3\sqrt3}{4}r_c,\ \frac{\sqrt3}{2}r_c\right].
\]
With $r_\beta=\sqrt{1-\beta^2}$, the formulas of the lemma become
\begin{equation}
\begin{aligned}
 \theta_3(\beta,\gamma)&=c-(1+c)\frac{\gamma^2}{r_\beta^2},
 &C_3(\beta,\gamma)&=\frac{r_\beta}{2}+\tau\frac{\beta\gamma}{r_\beta},
 &h_3(\beta,\gamma)&=-\frac\beta2+\tau\gamma,\\
 \chi_{3,\omega}&=\sqrt{\frac{(\theta_3+\omega)(\Sigma_{2,\omega}+\omega)}{1+c}},
 &\Lambda_{3,\omega}(\beta,\gamma)
 &=\frac{\Lambda_{2,\omega}C_3+\frac32-\tau\chi_{3,\omega}}
        {\sqrt{1-h_3^2}}.
\end{aligned}
\label{eq:level-three-threshold}
\end{equation}
On the box $B_c\times\Gamma_c$,
\[
 -\frac18-\frac32r_c\leq h_3\leq-\frac18+\frac34r_c,
\]
and since $r_c<1/4$, in particular $h_3>-1/2$.

\paragraph{Level three.}
The level-three endpoints have root correlation $h_3+o(1)$.  When their
candidate set is large, it is rounded with the vector coloring of
Lemma~\ref{lem:bhl-negative-coloring} if $h_3\leq0$ and with the strict
vector $3$-coloring if $h_3\geq0$.  The KMS loss exponent is
\begin{equation}
 \pi_3(h)=
 \begin{cases}
  \dfrac{1+2h}{3-6h},&h\leq0,\\[4pt]
  \dfrac13,&h\geq0,
 \end{cases}
 \label{eq:third-penalty}
\end{equation}
which is continuous.  Rounding gives $D^{A_0-o(1)}$ if the cardinality
coefficient is at least
\[
 \mu_3^*(h)=K_c\bigl(A_0+\pi_3(h)\bigr).
\]
The cutoff coefficient at level three is
\[
 a_{3,\omega}(\beta,\gamma)
 =\max\left\{\mu_3^*(h_3(\beta,\gamma)),\ \Lambda_{3,\omega}(\beta,\gamma)^2\right\},
\]
where the second term ensures that the surplus in the small case is
nonnegative.  In that case the cardinality coefficient is
$\mu_3=a_{3,\omega}+\omega$, and the surplus is
\begin{equation}
 \Sigma_{3,\omega}(\beta,\gamma)
 =a_{3,\omega}+\omega-\Lambda_{3,\omega}^2
 =\bigl(\mu_3^*(h_3)-\Lambda_{3,\omega}^2\bigr)_++\omega .
 \label{eq:sigma-three}
\end{equation}

\paragraph{Transition to level four.}
For a level-three endpoint $\ell$ and a next edge $\{\ell,m\}$, the
correlation $\eta_4=\langle v_i,v_{\ell m}\rangle$ has no prescribed range
in Proposition~\ref{prop:structural-input}.  Lemma~\ref{lem:discard-overcorrelated}
will let us assume $|\eta_4|\leq\rho_cr_{h_3}+o(1)$, where
$r_{h}=\sqrt{1-h^2}$, and we parametrize
\[
 \eta_4=\xi\rho_cr_{h_3},
 \qquad -1\leq\xi\leq1 .
\]
Since $(1+c)\rho_c^2=c$, the formulas of the lemma become
\begin{equation}
\begin{aligned}
 \theta_4(\xi)&=c(1-\xi^2),
 &C_4(h_3,\xi)&=\frac{r_{h_3}}{2}+\tau h_3\xi\rho_c,
 &h_4(h_3,\xi)&=-\frac{h_3}2+\tau\xi\rho_cr_{h_3},\\
 \chi_{4,\omega}&=\sqrt{\frac{(\theta_4+\omega)(\Sigma_{3,\omega}+\omega)}{1+c}},
 &\Lambda_{4,\omega}(\beta,\gamma,\xi)
 &=\frac{\Lambda_{3,\omega}C_4+\frac32-\tau\chi_{4,\omega}}
        {\sqrt{1-h_4^2}}.
\end{aligned}
\label{eq:level-four-threshold}
\end{equation}
The fourth level is used only for counting: a cover at threshold
$\Lambda_{4,\omega}t-o(t)$ needs at least
$\exp(\Lambda_{4,\omega}^2t^2/2-o(t^2))$ distinct endpoints, and this will
exceed $n$.

\paragraph{Summary.}
The following table lists the data at $\omega=0$.  The type box for the
last two transitions is
\begin{equation}
 \mathcal B_c=B_c\times\Gamma_c\times[-1,1].
 \label{eq:certificate-box}
\end{equation}
\begin{center}
\small
\renewcommand{\arraystretch}{1.25}
\begin{tabular}{c|c|c|c|c|c}
Level $k$ & $h_k$ & $\Lambda_k$ & $\mu_k$ (small case) & $\eta_{k+1}$ &
$\theta_{k+1}$\\
\hline
1 & $-\tfrac12$ & $\sqrt3$ & $K_c$ & $\tau\alpha$ &
 $c-(1+c)\alpha^2$\\
2 & $\beta=\tfrac14+\tfrac34\alpha$ & $\Lambda_2(\beta)$ & $\mu_2^*$ &
 $\gamma$ & $c-(1+c)\gamma^2/r_\beta^2$\\
3 & $h_3=-\tfrac\beta2+\tau\gamma$ & $\Lambda_3(\beta,\gamma)$ &
 $\max\{\mu_3^*(h_3),\Lambda_3^2\}$ & $\xi\rho_cr_{h_3}$ &
 $c(1-\xi^2)$\\
4 & $h_4$ & $\Lambda_4(\beta,\gamma,\xi)$ & & &
\end{tabular}
\end{center}

\subsection{The numerical inequalities}
\label{subsec:margin}

The argument needs the side conditions of
Lemma~\ref{lem:recursive-cover-step} at each transition, two exact
comparisons, and the fourth-level counting inequality.  At $\omega=0$
these are verified by the computer-assisted certificate of
Appendix~\ref{sec:numerics}, summarized in
Proposition~\ref{prop:certificate}.  The next lemma transfers them to a
fixed positive margin.

\begin{lemma}[Numerical inequalities with margin]
\label{lem:margin}
There is a margin $\omega_0\in(0,1/100]$ such that, for every
$\omega\in[0,\omega_0]$ and every $(\beta,\gamma,\xi)\in\mathcal B_c$, the
following hold.
\begin{enumerate}
\item[(i)] The side conditions of Lemma~\ref{lem:recursive-cover-step}
 hold at all three transitions: for the transition to level two,
 $\theta_2\geq0$, $\Sigma_1\geq0$, $C_2\geq0$, $|\beta|<1$, and
 $\Lambda_1C_2-\tau\chi_{2,\omega}>0$; for the transition to level three,
 $\theta_3>0$, $\Sigma_{2,\omega}>0$, $C_3>0$, $|h_3|<1$, and
 $\Lambda_{2,\omega}C_3-\tau\chi_{3,\omega}>0$; and for the transition to
 level four, $\theta_4\geq0$, $\Sigma_{3,\omega}\geq0$, $C_4>0$,
 $|h_4|<1$, and $\Lambda_{3,\omega}C_4-\tau\chi_{4,\omega}>0$.
\item[(ii)] $\mu_2^*/K_c>A_0>R(\widehat c)$.
\item[(iii)] $\Lambda_{4,\omega}(\beta,\gamma,\xi)^2/K_c>\Xi(\widehat c)$,
 with a positive gap independent of $(\beta,\gamma,\xi)$ and $\omega$.
\end{enumerate}
\end{lemma}

\begin{proof}
The conditions for the first transition are checked by hand.  On
$0\leq\alpha\leq r_c<1/5$ we have $\theta_2\geq c-(1+c)r_c^2>0$,
$\Sigma_1=3c>0$, $C_2>0$, and $\beta<2/5$.  For $\omega\leq1/100$,
\[
 \tau\chi_{2,\omega}
 \leq\tau\sqrt{\frac{(c+\tfrac1{100})(3c+\tfrac1{100})}{1+c}}
 <0.30
 <0.62
 <\frac34(1-r_c)
 \leq\Lambda_1C_2 ,
\]
using $\Lambda_1C_2=\frac34(1-\alpha)$.

Part~(ii) is an exact rational comparison and does not involve $\omega$;
it is verified in Appendix~\ref{sec:numerics}.

For the remaining conditions, observe that $\theta_3$, $C_3$, $h_3$,
$\theta_4$, $C_4$, and $h_4$ do not depend on $\omega$, and that
$\Lambda_{2,\omega}$, $\Sigma_{2,\omega}$, $\Lambda_{3,\omega}$,
$\Sigma_{3,\omega}$, and $\Lambda_{4,\omega}$ are built from them by sums,
products, quotients with denominators bounded away from zero on
$\mathcal B_c$, positive parts, and square roots of nonnegative
quantities.  Hence all quantities in (i) and (iii) are continuous
functions of $(\beta,\gamma,\xi,\omega)$ on the compact set
$\mathcal B_c\times[0,1/100]$.  At $\omega=0$,
Proposition~\ref{prop:certificate} states that every strict inequality in
(i) and (iii) holds on all of $\mathcal B_c$; by compactness, each holds
with a positive minimum margin.  By uniform continuity, there is
$\omega_0\in(0,1/100]$ such that, for $\omega\leq\omega_0$, every one of
these finitely many functions changes by less than half of its margin.
The nonstrict conditions $\theta_4\geq0$ and $\Sigma_{3,\omega}\geq0$
hold by definition.
\end{proof}

From now on, $\omega_0$ denotes the margin of Lemma~\ref{lem:margin}, and
the cutoff $N_t(\mu)$ of~\eqref{eq:cutoff} is defined with
$\omega=\omega_0$.

\subsection{The bounded-degree guarantee}

\begin{proposition}[Progress after a failed rounding]
\label{prop:failed-rounding-progress}
Let $G$ be an $n$-vertex $3$-colorable graph equipped with a solution to
the three-round relaxation of Section~\ref{subsec:overview-sdp}.  Suppose
$\Delta(G)\leq D$, and choose $t$ by
\[
 \tail(t)=D^{-1/[3(1+\widehat c)]},
 \qquad
 D=n^{(3+3\widehat c)/(5+3\widehat c)} .
\]
If the modified KMS rounding fails, then the algorithm described at the
end of this section finds, in randomized polynomial time and with high
probability, an independent set of size $D^{A_0-o(1)}$.
\end{proposition}

Here is the outline.  Proposition~\ref{prop:structural-input} provides the
level-one packing, the next-edge packings for the first two transitions,
and a next-edge packing for every level-three endpoint.  Before each
transition, Lemma~\ref{lem:common-bucket} restricts the relevant packings
to common buckets.  Lemma~\ref{lem:recursive-cover-step} propagates the
cover to level two and then to level three.  At each of these levels, the
good endpoints are contained in a candidate set that the algorithm tries,
and Lemma~\ref{lem:cardinality-split} gives two cases: either the
candidate set is large and KMS rounding finds the desired independent set,
or the good set is small and the recursion continues.  If both levels fall
into the small case, one last application of the recursion produces a
fourth-level cover, which by Lemma~\ref{lem:margin}(iii) would require more
than $n$ distinct endpoints.  Hence one of the two rounding steps succeeds.

\begin{proof}
Throughout, $t=\Theta(\sqrt{\log n})\to\infty$.  By
\eqref{eq:tail-log-uniform},
\begin{equation}
 \log D=3(1+\widehat c)\log\frac1{\tail(t)}
 =3(1+\widehat c)\left(\frac{t^2}2+O(\log t)\right),
 \label{eq:degree-threshold-relation}
\end{equation}
so $\log D=\Theta(t^2)$, and since $\widehat c<c$,
$\log D\leq K_ct^2/2$ for all large $n$.  These are the degree hypotheses
of Lemmas~\ref{lem:recursive-cover-step}, \ref{lem:discard-overcorrelated},
and~\ref{lem:cardinality-split}.  All representatives chosen below lie in
the compact box $\mathcal B_c$, on which every denominator in the
recursion is bounded away from zero, so the $o(\cdot)$ terms in the
conclusions of these lemmas are uniform in the representatives.

\medskip
\noindent\textbf{From level one to level two.}
Apply Proposition~\ref{prop:structural-input} and fix the root $i$ and the
sets $J$, $J_j$, $L_k$, $M_\ell$ that it supplies.  By part~(a), the
vectors $z_j=v_{ij}$, $j\in J$, carry a subexponential packing at
threshold $\Lambda_1t=\sqrt3t$, and $|J|\leq D\leq\exp(K_ct^2/2)$, so
$\mu_1=K_c$.  By part~(b), every $j\in J$ carries the packing of next-edge
directions $\{v_{jk}:k\in J_j\}$ at threshold $\sqrt3t$ with
$|J_j|\leq D$.

Label every $k\in J_j$ by the bucket $B\in\mathcal P_t$ containing
$\beta_k=\langle v_i,v_k\rangle$, which by~\eqref{eq:alpha-beta} depends
only on $k$.  Apply Lemma~\ref{lem:common-bucket} with a trivial outer
label and these inner labels.  It selects one bucket $B$ and restricts the
outer packing to a subset $J^B\subseteq J$ and each inner packing to
$J_j^B=\{k\in J_j:\beta_k\in B\}$, with all masses reduced by a factor at
most $Q$; they remain subexponential.  Part~(b) of the proposition and
\eqref{eq:alpha-beta} give $\beta_k\in[1/4-o(1),\,1/4+3r_c/4+o(1)]$, so
$B$ meets a $o(1)$-neighborhood of $B_c$.  Let $\beta$ be the point of
$B_c$ nearest to $B$, and put $\alpha=(4\beta-1)/3\in[0,r_c]$.  Then
$\beta_k=\beta+o(1)$ and $\langle v_i,v_{jk}\rangle=\tau\alpha+o(1)$
uniformly over the retained pairs, because the buckets have width
$O(1/\log t)$.

Apply Lemma~\ref{lem:recursive-cover-step} with margin $\omega_0$, with
$S^\circ=J^B$, $h=-1/2$, $\Lambda=\sqrt3$, $\mu=K_c$, and
$\eta=\tau\alpha$.  Its side conditions are Lemma~\ref{lem:margin}(i), and
its formulas are~\eqref{eq:level-two-threshold}.  It produces the set of
good level-two endpoints
\[
 S_2^\circ=\bigcup_{j\in J^B}J_j^B,
\]
whose vectors $z_k$ form a subexponential cover at threshold
$\Lambda_{2,\omega_0}(\beta)t-o(t)$, with $h_k=\beta_k\in B$ for every
$k\in S_2^\circ$.

For every root $i$ and bucket $B$, define the candidate set
\[
 \widehat S_2(i,B)=
 \{k\ne i:\text{there is a walk }i\to j\to k\text{ in }G
       \text{ with }\langle v_i,v_k\rangle\in B\},
\]
where repeated endpoints are listed once.  For the selected $B$,
$S_2^\circ\subseteq\widehat S_2(i,B)$; the root is excluded because
$h_i=1\notin B$ for all large $t$.

\medskip
\noindent\textbf{The level-two split.}
Put $\vartheta=\min\{1/4,\min_{k\in\widehat S_2(i,B)}\langle v_i,v_k\rangle\}$.
Every vertex of the candidate set has root correlation in $B$, so
$\vartheta=1/4-o(1)$, and eventually $1/16\leq\vartheta\leq1/4$.
Lemma~\ref{lem:bhl-positive-coloring} gives a vector $\kappa$-coloring of
$G[\widehat S_2(i,B)]$ with
\[
 \kappa=\frac{4+8\vartheta}{1+8\vartheta}=2+o(1),
 \qquad
 1-\frac2\kappa=o(1).
\]
Apply Lemma~\ref{lem:cardinality-split} with $\widehat S=\widehat S_2(i,B)$,
$S^\circ=S_2^\circ$, $\mu=\mu_2^*$, $\pi_t=0$, and $K=3$.  In its first
case, KMS rounding finds an independent set of size
$D^{\mu_2^*/K_c-o(1)}\geq D^{A_0}$ for large $n$, because
$\mu_2^*/K_c>A_0$ by Lemma~\ref{lem:margin}(ii), and we are done.  We may
therefore assume its second case:
\[
 \log|S_2^\circ|<\frac{\mu_2^*+\omega_0}{2}t^2 .
\]

\medskip
\noindent\textbf{From level two to level three.}
By the remark after the definition of packing measures, the cover on
$S_2^\circ$ carries a packing measure of the same subexponential mass.
By Proposition~\ref{prop:structural-input}(c), every $k\in S_2^\circ\subseteq W$
carries the packing $\{v_{k\ell}:\ell\in L_k\}$ at threshold $\sqrt3t$
with $|L_k|\leq D$, and $\gamma_{k\ell}=\langle v_i,v_{k\ell}\rangle$ lies
within $o(1)$ of $\Gamma_c$.  Label each $\ell\in L_k$ by the bucket
$B'\in\mathcal P_t$ containing $\gamma_{k\ell}$ and apply
Lemma~\ref{lem:common-bucket}.  It selects one bucket $B'$, restricts the
outer packing to $S_2^{\circ,B'}\subseteq S_2^\circ$ and each inner packing
to $L_k^{B'}=\{\ell\in L_k:\gamma_{k\ell}\in B'\}$, and keeps all masses
subexponential.  Let $\gamma$ be the point of $\Gamma_c$ nearest to $B'$;
then $\gamma_{k\ell}=\gamma+o(1)$ uniformly.

Apply Lemma~\ref{lem:recursive-cover-step} with margin $\omega_0$, with
$S^\circ=S_2^{\circ,B'}$, $h=\beta$, $\Lambda=\Lambda_{2,\omega_0}(\beta)$,
$\mu=\mu_2^*+\omega_0$, and $\eta=\gamma$.  The cardinality hypothesis is
the second case above, the side conditions are Lemma~\ref{lem:margin}(i),
and the formulas are~\eqref{eq:sigma-two}
and~\eqref{eq:level-three-threshold}.  It produces the good level-three
endpoints
\[
 S_3^\circ=\bigcup_{k\in S_2^{\circ,B'}}L_k^{B'},
\]
whose vectors $z_\ell$ form a subexponential cover at threshold
$\Lambda_{3,\omega_0}(\beta,\gamma)t-o(t)$, with
$h_\ell=h_3(\beta,\gamma)+o(1)$.

For every root $i$ and buckets $B,B'$, define the candidate set
\[
 \widehat S_3(i,B,B')=
 \{\ell\ne i:\text{there is a walk }i\to j\to k\to\ell\text{ in }G
   \text{ with }\langle v_i,v_k\rangle\in B,\
   \langle v_i,v_{k\ell}\rangle\in B'\},
\]
with repeated endpoints listed once.  For the selected buckets,
$S_3^\circ\subseteq\widehat S_3(i,B,B')$.  Moreover, the exact identity
$\langle v_i,v_\ell\rangle=-\langle v_i,v_k\rangle/2+\tau\langle v_i,v_{k\ell}\rangle$
from~\eqref{eq:step-endpoint-correlation} shows that every vertex of the
candidate set has root correlation $h_3(\beta,\gamma)+O(1/\log t)$,
whichever walk reached it.

\medskip
\noindent\textbf{The level-three split.}
Put $H=\max_{\ell\in\widehat S_3(i,B,B')}\langle v_i,v_\ell\rangle
=h_3(\beta,\gamma)+o(1)$.  Since $h_3>-1/2$ on the box, $H>-1/2$ for all
large $t$.  If $H\leq0$, Lemma~\ref{lem:bhl-negative-coloring} with $h=H$
gives a vector $\kappa$-coloring of $G[\widehat S_3(i,B,B')]$ with
$\kappa=(3-6H)/(1-4H)\in[2,3]$ and $1-2/\kappa=(1+2H)/(3-6H)=\pi_3(H)$.
If $H>0$, the strict vector $3$-coloring is a vector $3$-coloring with
$1-2/\kappa=1/3=\pi_3(H)$.  In both cases $1-2/\kappa=\pi_3(H)
=\pi_3(h_3(\beta,\gamma))+o(1)$, because $\pi_3$ is continuous.  Apply
Lemma~\ref{lem:cardinality-split} with $\widehat S=\widehat S_3(i,B,B')$,
$S^\circ=S_3^\circ$, $\mu=a_{3,\omega_0}(\beta,\gamma)$,
$\pi_t=\pi_3(h_3(\beta,\gamma))$, and $K=3$.  In its first case, KMS
rounding finds an independent set of size at least
\[
 D^{a_{3,\omega_0}/K_c-\pi_3(h_3)-o(1)}
 \geq D^{\mu_3^*(h_3)/K_c-\pi_3(h_3)-o(1)}
 =D^{A_0-o(1)},
\]
and we are done.  We may therefore assume its second case:
\[
 \log|S_3^\circ|<\frac{a_{3,\omega_0}(\beta,\gamma)+\omega_0}{2}t^2,
\]
so the cardinality coefficient for the last transition is
$a_{3,\omega_0}+\omega_0$ and the surplus is
$\Sigma_{3,\omega_0}(\beta,\gamma)$ from~\eqref{eq:sigma-three}.

\medskip
\noindent\textbf{From level three to level four.}
By Proposition~\ref{prop:structural-input}(d), every $\ell\in S_3^\circ$
carries the packing $\{v_{\ell m}:m\in M_\ell\}$ at threshold $\sqrt3t$
with $|M_\ell|\leq D$ and subexponential mass $\delta$, the same for all
$\ell$.  Since $|h_\ell|\leq|h_3|+o(1)$ is bounded away from one on the
box, $r_\ell\geq r_{\min}$ for a fixed $r_{\min}>0$.
Lemma~\ref{lem:discard-overcorrelated} gives $\epsilon_t\to0$ such that
removing from every $M_\ell$ the indices $m$ with
$|\eta_{\ell m}|\geq\rho_cr_\ell+\epsilon_t$ leaves mass at least
$3\delta/4$.  For the retained pairs put
$\xi_{\ell m}=\eta_{\ell m}/(\rho_cr_\ell)$, so that
$|\xi_{\ell m}|\leq1+\epsilon_t/(\rho_cr_{\min})=1+o(1)$.  Label each
retained $m$ by the bucket of a fixed partition of $[-2,2]$ into $Q$
intervals containing $\xi_{\ell m}$, and apply
Lemma~\ref{lem:common-bucket}.  It selects one bucket, restricts the outer
packing on $S_3^\circ$ and the inner packings, and keeps all masses
subexponential.  Let $\xi\in[-1,1]$ be the point nearest to that bucket.
Then $\xi_{\ell m}=\xi+o(1)$, and since $r_\ell=r_{h_3}+o(1)$,
\[
 \eta_{\ell m}=\xi\rho_cr_{h_3}+o(1)
\]
uniformly over the retained pairs.

Apply Lemma~\ref{lem:recursive-cover-step} with margin $\omega_0$, with
$S^\circ$ the restricted level-three set, $h=h_3(\beta,\gamma)$,
$\Lambda=\Lambda_{3,\omega_0}(\beta,\gamma)$,
$\mu=a_{3,\omega_0}(\beta,\gamma)+\omega_0$, and
$\eta=\xi\rho_cr_{h_3}$.  The side conditions are
Lemma~\ref{lem:margin}(i), and the formulas are
\eqref{eq:level-four-threshold}.  It produces a set $M^\circ$ of
fourth-level endpoints whose vectors $z_m$ form a cover at threshold
$\Lambda_{4,\omega_0}(\beta,\gamma,\xi)t-o(t)$ with subexponential
probability $\delta_4$.

\medskip
\noindent\textbf{Counting the fourth level.}
The union bound and~\eqref{eq:tail-log-uniform} give
\[
 \delta_4
 \leq|M^\circ|\,\tail\bigl(\Lambda_{4,\omega_0}t-o(t)\bigr),
 \qquad\text{hence}\qquad
 \log|M^\circ|
 \geq\frac{\Lambda_{4,\omega_0}^2}{2}t^2-o(t^2)
 \geq\frac{\Lambda_{4,\omega_0}^2}{K_c}\log D-o(\log D),
\]
using $\log D\leq K_ct^2/2$ and $\log D=\Theta(t^2)$.  By
Lemma~\ref{lem:margin}(iii), $\Lambda_{4,\omega_0}^2/K_c\geq\Xi(\widehat c)+g$
for a fixed $g>0$, so for all large $n$
\[
 \log|M^\circ|>\Xi(\widehat c)\log D=\log n,
\]
because $D=n^{(3+3\widehat c)/(5+3\widehat c)}$ means $n=D^{\Xi(\widehat c)}$.
But $M^\circ$ is a set of distinct vertices of $G$, so $|M^\circ|\leq n$.
This contradiction shows that the second case cannot occur at both levels,
so one of the two rounding steps finds an independent set of size
$D^{A_0-o(1)}$.

\medskip
\noindent\textbf{The algorithm.}
The case split above is used only in the proof.  The algorithm does not
test whether the rounding fails, and it does not construct the packing
measures, the pruned subgraph, the buckets selected in the proof, or the
good endpoint sets.  It performs the following operations.  Whenever a
displayed vector must be normalized, the algorithm checks that its squared
norm is positive and otherwise skips that candidate; the candidates
selected in the proof have nonzero vectors by
Lemmas~\ref{lem:bhl-negative-coloring} and~\ref{lem:bhl-positive-coloring}.
\begin{enumerate}
\item Solve the three-round relaxation.  Repeat the modified KMS rounding
on the whole graph at threshold $t$, retaining every independent set
produced.

\item For every root $i$ and every bucket $B\in\mathcal P_t$, enumerate
all walks $i\to j\to k$, merge repeated endpoints, and construct
$\widehat S_2(i,B)$.  Skip an empty set.  Otherwise put
$\vartheta_{i,B}=\min\{1/4,\min_{k}\langle v_i,v_k\rangle\}$ over
$k\in\widehat S_2(i,B)$; if $\vartheta_{i,B}<1/16$, skip this candidate.
Otherwise normalize the vectors
\[
 y_{\{(i,\mathsf R),(k,\mathsf R)\}}
 -y_{\{(i,\mathsf R),(k,\mathsf G)\}}
 -y_{\{(i,\mathsf R),(k,\mathsf B)\}},
 \qquad k\in\widehat S_2(i,B),
\]
of Lemma~\ref{lem:bhl-positive-coloring}, and run KMS rounding on
$G[\widehat S_2(i,B)]$ with this vector coloring.

\item For every root $i$ and every ordered pair of buckets
$(B,B')\in\mathcal P_t^2$, enumerate all walks $i\to j\to k\to\ell$, merge
repeated endpoints, and construct $\widehat S_3(i,B,B')$.  Skip an empty
set.  Otherwise let $H=\max_\ell\langle v_i,v_\ell\rangle$ over
$\ell\in\widehat S_3(i,B,B')$; if $H<-1/2$, skip this candidate.  If
$-1/2\leq H\leq0$, normalize the vectors
\[
 y_{\{(i,\mathsf R),(\ell,\mathsf G)\}}
 -y_{\{(i,\mathsf R),(\ell,\mathsf B)\}},
 \qquad \ell\in\widehat S_3(i,B,B'),
\]
of Lemma~\ref{lem:bhl-negative-coloring}, which form a vector
$(3-6H)/(1-4H)$-coloring; if $H>0$, use the strict vector $3$-coloring.
Run KMS rounding on $G[\widehat S_3(i,B,B')]$ with the chosen vector
coloring.

\item Return the largest independent set found.
\end{enumerate}
All calls are made regardless of which case occurs in the analysis.  If
the rounding succeeds, the first step finds a large independent set.  If
it fails, the analysis above shows that one candidate from the second or
third step has the required expected size.  There are at most
$1+nQ+nQ^2$ candidate families, each built by enumerating walks of length
at most three, so the total running time is polynomial.

Finally, suppose that one KMS call has expected output size at least
$M\geq1$; its output always has size at most $n$, so
$\Pr[|I|\geq M/2]\geq M/(2n)\geq1/(2n)$.  Repeating every call $O(n\log n)$
times with a sufficiently large constant makes the probability that any
qualifying call misses half of its expected size smaller than any
prescribed inverse power of $n$, and a union bound over the
$1+nQ+nQ^2$ families gives the high-probability guarantee.
\end{proof}

\section{Proof of the coloring theorem}
\label{sec:coloring-proof}

We first state the two known results with which we combine the
bounded-degree guarantee.  Following \citet*{BHL26}, an algorithm
\emph{makes progress toward an $f(n)$-coloring} if it performs at least
one of the following operations: it finds an independent set of size
$\Omega(n/f(n))$; it finds an independent set $S$ with $|N(S)|=O(f(n)|S|)$;
or it identifies two distinct vertices that must receive the same color in
every proper $3$-coloring.  The combination theorem below shows that any
one of these operations is enough for the coloring algorithm.

The first result is due to \citet*{KTY24} and is restated as
Theorem~1.2 of \citet*{BHL26}.  We use $d$ for the minimum-degree bound
to avoid conflict with our notation $\Delta(G)$ for the maximum degree.

\begin{mainresultbox}
\begin{theorem}[Dense-graph progress theorem \citep{KTY24}]
\label{thm:bhl-dense-progress}
Let $G$ be a $3$-colorable graph on $n$ vertices with minimum degree
$d>n^{1/2}$.  In polynomial time, one can make progress toward a
$k$-coloring for some
\[
 k=n^{o(1)}\sqrt{\frac nd}.
\]
\end{theorem}
\end{mainresultbox}

The second result uses the dense--sparse framework of \citet*{BK97}; the
form below is Proposition~7.1 of \citet*{KT17}, restated as Theorem~2.6 of
\citet*{BHL26}.

\begin{mainresultbox}
\begin{theorem}[Dense--sparse combination theorem \citep{KT17,BHL26}]
\label{thm:bhl-dense-sparse}
Let $\alpha,\beta\in(0,1)$.  Suppose that a polynomial-time algorithm can
make progress toward an $\widetilde O(n^\alpha)$-coloring of every
$3$-colorable $n$-vertex graph $G$ whenever either
$\Delta(G)\leq n^\beta$ or the minimum degree of $G$ is at least $n^\beta$.
Then a proper $\widetilde O(n^\alpha)$-coloring can be found in polynomial
time.
\end{theorem}
\end{mainresultbox}

Proposition~\ref{prop:failed-rounding-progress}, together with direct
Gaussian rounding, will satisfy the low-maximum-degree hypothesis, and
Theorem~\ref{thm:bhl-dense-progress} handles the high-minimum-degree
hypothesis.  It remains to check that the two guarantees have matching
exponents.

\begin{proof}[Proof of Theorem~\ref{thm:final-coloring}]
Recall $\widehat c=2067/10000$ from~\eqref{eq:analytic-parameters}, and
set
\[
 \widehat\alpha=\frac1{5+3\widehat c},
 \qquad
 \widehat\beta=\frac{3+3\widehat c}{5+3\widehat c},
 \qquad
 D=n^{\widehat\beta}.
\]
These definitions were chosen so that
\begin{equation}
 \frac{1-\widehat\beta}{2}
 =\frac{\widehat\beta}{3(1+\widehat c)}
 =\widehat\alpha,
 \qquad
 \widehat\beta\,R(\widehat c)=1-\widehat\alpha,
 \label{eq:parameter-identities}
\end{equation}
with $R(x)=(4+3x)/(3+3x)$ as in Section~\ref{subsec:constants}.  The first
identity balances the dense and sparse cases and sets the Gaussian
threshold below.  The second says that an independent set of size
$D^{R(\widehat c)-o(1)}$ has the target size $n^{1-\widehat\alpha-o(1)}$.

There is one minor exponent issue to handle before applying
Theorem~\ref{thm:bhl-dense-sparse}.  Its hypothesis allows a
polylogarithmic loss, whereas the routines below have an $n^{o(1)}$ loss.
Set
\[
 \alpha_0=0.17794,
 \qquad
 \alpha_*=\frac{\widehat\alpha+\alpha_0}{2}.
\]
Since $\widehat\alpha=10000/56201=0.1779327\ldots<\alpha_0$, we have
$\widehat\alpha<\alpha_*<\alpha_0$ with a fixed positive gap on each side.
We apply the combination theorem with $\alpha=\alpha_*$ and
$\beta=\widehat\beta$.  It suffices to make progress toward an
$O(n^{\alpha_*})$-coloring on graphs with maximum degree at most $D$ and
on graphs with minimum degree at least $D$.

\paragraph{The dense case.}
If the minimum degree is $d\geq D$, then $D>n^{1/2}$, and
Theorem~\ref{thm:bhl-dense-progress} makes progress toward a $k$-coloring
with
\[
 k=n^{o(1)}\sqrt{\frac nd}
  \leq n^{o(1)}\sqrt{\frac nD}
  =n^{\widehat\alpha+o(1)}
  \leq n^{\alpha_*}
\]
for all large $n$, using $(1-\widehat\beta)/2=\widehat\alpha$ and
$\widehat\alpha<\alpha_*$.

\paragraph{The sparse case.}
Suppose $\Delta(G)\leq D$.  Solve the three-round relaxation of
Section~\ref{subsec:overview-sdp}, and choose the Gaussian threshold $t$ by
\[
 \tail(t)=D^{-1/[3(1+\widehat c)]}=n^{-\widehat\alpha},
\]
where the second equality is the first identity
in~\eqref{eq:parameter-identities}.  In particular
$t=\Theta(\sqrt{\log n})\to\infty$.

If the modified KMS rounding succeeds, the argument of
Section~\ref{subsec:overview-rounding} finds, after polynomially many
repetitions, an independent set of size
$\Omega(n\tail(t))=\Omega(n^{1-\widehat\alpha})$, which is at least the
$\Omega(n^{1-\alpha_*})$ required for progress toward an
$O(n^{\alpha_*})$-coloring.

If the rounding fails, Proposition~\ref{prop:failed-rounding-progress}
finds, with high probability, an independent set of size
$D^{A_0-o(1)}$.  Since $A_0>R(\widehat c)$ by Lemma~\ref{lem:margin}(ii),
the second identity in~\eqref{eq:parameter-identities} gives
\[
 D^{A_0-o(1)}
 \geq D^{R(\widehat c)-o(1)}
 =n^{1-\widehat\alpha-o(1)}
 \geq n^{1-\alpha_*}
\]
for all large $n$.  We have therefore made the required progress in every
case, and Theorem~\ref{thm:bhl-dense-sparse} gives a proper
$\widetilde O(n^{\alpha_*})$-coloring.

\paragraph{The final exponent.}
Since $\alpha_*<\alpha_0=0.17794$ by a fixed positive amount, the
polylogarithmic factor in $\widetilde O(n^{\alpha_*})$ is at most
$n^{\alpha_0-\alpha_*}$ for all large $n$, and increasing the implicit
constant handles the finitely many smaller values of $n$.  The number of
colors is therefore $O(n^{0.17794})$.  Amplifying each randomized progress
routine to inverse-polynomial failure probability and taking a union bound
over the polynomially many calls gives the claimed randomized algorithm.
\end{proof}

\appendix

\section{The numerical certificate}
\label{sec:numerics}

This appendix describes the finite computer-assisted verification used by
Lemma~\ref{lem:margin}.  It verifies, at margin $\omega=0$, the side
conditions of the recursion for the transitions to levels three and four,
the two exact comparisons, and the fourth-level counting inequality.  The
statement is the following.

\begin{proposition}[Certified inequalities]
\label{prop:certificate}
At the parameters~\eqref{eq:analytic-parameters}, the following hold for
every $(\beta,\gamma,\xi)$ in the box $\mathcal B_c$
of~\eqref{eq:certificate-box}, with the quantities of
Section~\ref{subsec:constants} evaluated at $\omega=0$.
\begin{enumerate}
\item[(i)] $\theta_3>0$, $\Sigma_2>0$, $C_3>0$, $|h_3|<1$, $\Lambda_3>0$,
 and $\Lambda_2C_3-\tau\sqrt{\theta_3\Sigma_2/(1+c)}>0$; and
 $C_4>0$, $|h_4|<1$, and $\Lambda_3C_4-\tau\sqrt{\theta_4\Sigma_3/(1+c)}>0$.
\item[(ii)] $\mu_2^*/K_c>A_0>R(\widehat c)$.
\item[(iii)] $\Lambda_4(\beta,\gamma,\xi)^2/K_c>\Xi(\widehat c)$.
\end{enumerate}
\end{proposition}

\begin{proof}
Part~(ii) consists of the exact rational comparisons
\eqref{eq:certificate-level-two-progress}
and~\eqref{eq:certificate-progress-gap} below.  Parts~(i) and~(iii) are
proved by the interval computation described in the rest of this
appendix, whose certified lower bound~\eqref{eq:certificate-primary-bound}
exceeds $\Xi(\widehat c)$.
\end{proof}

\paragraph{What the computation does.}
The verification uses only the primitive rational parameters
\begin{equation}
\begin{aligned}
 c&=\frac{207}{1000},
 &\widehat c&=\frac{2067}{10000},\\
 \mu_2^*&=\frac{9248493504087891}{2000000000000000},
 &A_0&=\frac{3191}{2500}.
\end{aligned}
\label{eq:certificate-parameters}
\end{equation}
The role of the computation is narrow: it checks that the recursion
formulas remain well defined and satisfy the required strict inequalities
for every point of the box.  The proof does not rely on a sampled grid or
on the location suggested by a numerical optimizer.  Instead, interval
arithmetic with outward rounding covers the entire box: each lower endpoint
is rounded downward and each upper endpoint upward, so every computed
interval provably contains the exact range of the expression.  This is an
offline verification of fixed constants, not a step of the coloring
algorithm.

The verification uses $70$-digit decimal arithmetic with directed rounding.
The ancillary source file
\texttt{combined\_interval\_certificate.py}, distributed with this
manuscript, is self-contained and uses only the Python standard library.
From the directory containing that file, run
\begin{verbatim}
python3 combined_interval_certificate.py
\end{verbatim}
The run reported below used CPython~3.12.3 under WSL2.  The SHA--256 digest
of the source file is
\begin{center}
\ttfamily e9ac70c7b421b00aa482a7bcfa13ad9471aa5d869e1f0afa574c1c085952fa6a
\end{center}
The companion file
\texttt{combined\_interval\_certificate\_output.txt} contains the command,
exit status, interpreter and platform details, and the complete terminal
output from that run.  Its SHA--256 digest is
\begin{center}
\ttfamily e20cbf330da572d95bf21349e58a710abc44a82c3c488c01f6cb7f90de74ccb7
\end{center}
A successful run terminates only after every asserted enclosure and exact
rational comparison has passed; it prints the subdivision count, the
maximum depth, the certified lower endpoint, and the rational comparison
gaps.

The source file predates the notation of this paper in two places: it
calls the level-two coefficient $\Lambda_2(\beta)$ by the name
\texttt{eta}, and it calls the fourth-level parameter $\xi$ by the name
\texttt{z}.  The computation represents $\mu_2^*$ as
\[
 \frac{2291}{2000}\,\zeta^2,
 \qquad
 2.0091988\leq\zeta\leq2.0091990.
\]
The resulting interval bounds hold for every $\zeta$ in this closed
interval, so in particular for the rational endpoint
$\zeta=2009199/10^6$ used in~\eqref{eq:analytic-parameters}.  The
separate minimization that originally suggested this interval is not needed
by the proof.  The implementation also uses the relative variables
\[
 \frac{\Sigma_2(\beta)}{\Lambda_2(\beta)^2},
 \qquad
 \frac{\Sigma_3(\beta,\gamma)}{\Lambda_3(\beta,\gamma)^2},
\]
which is a reparametrization of the surplus formulas
\eqref{eq:sigma-two} and~\eqref{eq:sigma-three} at $\omega=0$; no
inequality is changed by it.

For every accepted subbox, the computation verifies
\[
\begin{aligned}
 0&<\frac{\Sigma_2}{\Lambda_2^2}<\frac15,
 &\theta_3&>0,\\
 0&\leq\frac{\Sigma_3}{\Lambda_3^2}<\frac1{10},
\end{aligned}
\]
together with positivity of the combined coefficients
\begin{align*}
 \Lambda_2C_3-\tau\sqrt{\frac{\theta_3\Sigma_2}{1+c}}&>0,
 \\[-2mm]
 \Lambda_3C_4-\tau\sqrt{\frac{\theta_4\Sigma_3}{1+c}}&>0,
\end{align*}
positivity of $\Lambda_3$, and the normalizations $|h_3|<1$ and
$|h_4|<1$.  Because $\Lambda_2>0$ and $\Lambda_3>0$, the two strict
combined inequalities also imply $C_3>0$ and $C_4>0$.  These are the
conditions in part~(i) of the proposition.  The two bounds $1/5$ and
$1/10$ are auxiliary range checks on the relative variables and are not
used by the proof.

\paragraph{The fourth-level bound.}
Most subboxes are settled by substituting the interval for each variable
directly into the formulas, rounding every lower endpoint down and every
upper endpoint up; we call this natural interval evaluation.  Near the
interior minimum, this direct bound can be too wide, so the verification
uses the mean-value inclusion
\[
 f(X)\in f(m)+\nabla f(X)\cdot(X-m),
 \qquad
 f(\beta,\gamma,\xi)=\frac{\Lambda_4(\beta,\gamma,\xi)^2}{K_c},
\]
where $X$ is a subdivision box, $m$ is its midpoint, and $\nabla f(X)$ is
an outward-rounded interval enclosure of the gradient on $X$.  This
mean-value step is used only on smooth interior subboxes.  At $h_3=0$, the
two formulas for the penalty $\pi_3$ agree and their one-sided derivatives
are finite, so the interval hull of those derivatives gives a valid
mean-value enclosure across that boundary.

The boundary $\Sigma_3=0$ is handled differently.  On every box where the
interval for $\mu_3^*(h_3)-\Lambda_3^2$ can be nonpositive, the verifier
does not differentiate the square root of its positive part.  It uses
natural interval evaluation and accepts the box only if the resulting lower
endpoint already proves the strict bound; otherwise it subdivides the box.
The same natural-only rule is used for boxes touching $\xi=\pm1$, so the
square root in $\theta_4$ is never differentiated there.  Assertions at
every accepted leaf check that these boxes followed the natural-only path,
and separate counters for a zero surplus, a box crossing zero, and a
positive surplus make this treatment visible in the terminal summary.

Here is the correspondence between the proposition and the checks in the
source file.
\begin{itemize}
\item The function \texttt{main} proves strict convexity of
$\alpha\mapsto\Lambda_2(1/4+3\alpha/4)$, isolates its critical point on
$0\leq\alpha\leq r_c$, and encloses its minimum.  This auxiliary check
documents how the interval for $\zeta$ was found; the proof itself uses
only the displayed enclosure for $\zeta$.
\item The functions \texttt{partition}, \texttt{audit\_beta\_slab}, and
\texttt{certify\_box} begin with the whole box $\mathcal B_c$ and replace
any unsettled box by two children whose union is the parent.  A successful
run therefore leaves no part of $\mathcal B_c$ unchecked.
\item The functions \texttt{third\_type}, \texttt{calibrated\_cutoff}, and
\texttt{fourth\_raw} evaluate the formulas
\eqref{eq:level-three-threshold}, \eqref{eq:sigma-three}, and
\eqref{eq:level-four-threshold} at $\omega=0$.  The Boolean value
\texttt{valid}, together with the guarded square-root and division
operations, checks every condition in part~(i).
\item The cutoff classification and the acceptance assertions enforce the
natural-only rule at $\Sigma_3=0$.  The comparison with
\texttt{RAW\_ENDPOINT\_CAP} proves the fourth-level bound on every accepted
leaf, and a final assertion checks the minimum over all leaves.
\item Exact \texttt{Fraction} assertions prove the rational comparisons in
\eqref{eq:certificate-level-two-progress} and
\eqref{eq:certificate-progress-gap}, as well as the final coloring exponent.
\end{itemize}
The interval computation visits $969{,}772$ boxes, reaches maximum depth
$17$, and proves
\begin{equation}
 \inf_{(\beta,\gamma,\xi)\in\mathcal B_c}
 \frac{\Lambda_4(\beta,\gamma,\xi)^2}{K_c}
 \geq1.5524720442257730964483.
 \label{eq:certificate-primary-bound}
\end{equation}
This bound is strictly greater than
\[
 \Xi(\widehat c)
 =\frac{5+3\widehat c}{3+3\widehat c}
 =\frac{56201}{36201}
 =1.5524709262175\ldots,
\]
so the certified separation is at least $1.1180082\times10^{-6}$.  This is
part~(iii).

\paragraph{The exact comparisons.}
The remaining comparisons are exact rational arithmetic.  First,
\begin{align}
 \frac{\mu_2^*}{K_c}
 &=\frac{3082831168029297}{2414000000000000},
 \notag\\
 \frac{\mu_2^*}{K_c}-A_0
 &=\frac{1601568029297}{2414000000000000}>0.
 \label{eq:certificate-level-two-progress}
\end{align}
Second,
\begin{align}
 R(\widehat c)
 &=\frac{4+3\widehat c}{3+3\widehat c}
 =\frac{46201}{36201},
 \notag\\
 A_0-R(\widehat c)
 &=\frac{14891}{90502500}>0.
 \label{eq:certificate-progress-gap}
\end{align}
This is part~(ii).

For reference, a compact transcription of the final terminal summary is
\begin{small}
\begin{verbatim}
certified raw fourth lower endpoint:
  1.552472044225773096448329261981444417504422522963453492483245775267468
boxes: 969772; accepted leaves: 485910; maximum depth: 17
h-branch leaves: negative 485653; positive 0; crossing 257
cutoff leaves: zero 32; crossing 1013; positive 484865
evaluation: natural 477090; centered 8820
mu2/K_c: 3082831168029297/2414000000000000
mu2/K_c - A0: 1601568029297/2414000000000000
A0 - R(c_hat): 14891/90502500
final exponent: 10000/56201
\end{verbatim}
\end{small}

\section{Proof of the KMS rounding theorem}
\label{sec:kms-proof}

We prove Theorem~\ref{thm:kms-extraction}.  The argument is the standard
one from \citet*{KMS98}; we include it to make explicit that the bound is
uniform in $\kappa$ near $2$.

\begin{proof}[Proof of Theorem~\ref{thm:kms-extraction}]
Set $\pi=1-2/\kappa\in[0,1]$.  If $\kappa=2$, the vectors at the endpoints
of every edge are antipodal, so thresholding a random Gaussian projection
at zero partitions the vertices into two independent sets almost surely;
the larger has size at least $m/2$.  We may also replace $D$ by
$\max\{1,\Delta(H)\}$, since doing so only strengthens the claimed bound.

Suppose that $\kappa>2$, and set
\[
 u=\sqrt{2\pi\log(2D)},
 \qquad
 a=\sqrt{\frac{2(\kappa-1)}{\kappa-2}}.
\]
For a standard Gaussian vector $g$, select a vertex $i$ when
$\langle g,v_i\rangle\geq u$.  Each vertex is selected with probability
$\tail(u)$.  If adjacent vertices $i,j$ are both selected, then the event
is impossible when $v_i+v_j=0$; otherwise
$\langle v_i,v_j\rangle\leq-1/(\kappa-1)$ gives
$\|v_i+v_j\|\leq\sqrt{2-2/(\kappa-1)}=2/a$, and hence
$\langle g,(v_i+v_j)/\|v_i+v_j\|\rangle\geq au$.  Thus that edge is
selected with probability at most $\tail(au)$.

For $y\geq x\geq0$, the Mills-ratio bound~\eqref{eq:gaussian-mills} gives
$\phi(x)/\tail(x)\geq x$, and integrating this hazard-rate bound yields
\[
 \frac{\tail(y)}{\tail(x)}
 \leq \exp\!\left(-\frac{y^2-x^2}{2}\right).
\]
Since $(a^2-1)u^2/2=\log(2D)$, it follows that
$\tail(au)\leq\tail(u)/(2D)$.  The expected number of selected vertices
minus selected edges is therefore at least
$m\tail(u)-mD\tail(u)/(2D)\geq m\tail(u)/2$.  Deleting one endpoint of
every selected edge leaves an independent set at least this large in
expectation.  For $u\leq1$ use $\tail(u)\geq\tail(1)$, while for $u>1$ the
lower bound in~\eqref{eq:gaussian-mills} gives
$\tail(u)\geq u\phi(u)/(1+u^2)$.  Thus, uniformly in $\pi\in[0,1]$,
\[
 \tail(u)
 =\Omega\!\left(
   \frac{D^{-\pi}}{\sqrt{\log(2D)}}
 \right).
\]
Because $D\leq m$ and $K$ is fixed, this expectation is an inverse-polynomial
fraction of $m$.  Polynomially many repetitions and selection of the largest
outcome therefore give the stated randomized guarantee.
\end{proof}

\bibliographystyle{plainnat}
\bibliography{references}

\end{document}